\documentclass[12pt]{article}
\usepackage{amsthm,amsmath,indentfirst,amssymb}
\usepackage{algorithm,algpseudocode}
\newtheorem{lemma}{Lemma}[section]
\newtheorem{theorem}[lemma]{Theorem}
\newtheorem{corollary}[lemma]{Corollary}
\newtheorem{definition}[lemma]{Definition}
\newtheorem{observation}[lemma]{Observation}

\begin{document}

\title{Approximation Algorithm for Fault-Tolerant Virtual Backbone in Wireless Sensor Networks}
\author{\footnotesize Jiao Zhou$^{1}$\quad Zhao Zhang$^1$\thanks{Corresponding author: hxhzz@sina.com.}\quad Xiaohui Huang$^1$ \quad Dingzhu Du$^2$\\
   {\it\small $^1$ College of Mathematics Physics and Information Engineering, Zhejiang Normal University}\\
    {\it\small Jinhua, Zhejiang, 321004, China}\\
    {\it\small $^2$ Department of Computer Science, University of Texas at Dallas}\\
    {\it\small Richardson, Texas, 75080, USA}}
\date{}
\maketitle
\begin{abstract}
%\boldmath
To save energy and alleviate interferences in a wireless sensor network, the usage of virtual backbone was proposed. Because of accidental damages or energy depletion, it is desirable to construct a fault tolerant virtual backbone, which can be modeled as a $k$-connected $m$-fold dominating set (abbreviated as $(k,m)$-CDS) in a graph. A node set $C\subseteq V(G)$ is a $(k,m)$-CDS of graph $G$ if every node in $V(G)\backslash C$ is adjacent with at least $m$ nodes in $C$ and the subgraph of $G$ induced by $C$ is $k$-connected. In this paper, we present an approximation algorithm for the minimum $(3,m)$-CDS problem with $m\geq3$. The performance ratio is at most $\gamma$, where $\gamma=\alpha+8+2\ln(2\alpha-6)$ for $\alpha\geq4$ and $\gamma=3\alpha+2\ln2$ for $\alpha<4$, and $\alpha$ is the performance ratio for the minimum $(2,m)$-CDS problem. Using currently best known value of $\alpha$, the performance ratio is $\ln\delta+o(\ln\delta)$, where $\delta$ is the maximum degree of the graph, which is asymptotically best possible in view of the non-approximability of the problem. This is the first performance-guaranteed algorithm for the minimum $(3,m)$-CDS problem on a general graph. Furthermore, applying our algorithm on a unit disk graph which models a homogeneous wireless sensor network, the performance ratio is less than 27, improving previous ratio 62.3 by a large amount for the $(3,m)$-CDS problem on a unit disk graph.
\end{abstract}
\section{Introduction}\label{sec1}

A wireless sensor network $($WSN$)$ consists of spatially distributed autonomous sensors to monitor physical or environmental condition, and to cooperatively pass the sensed data through the network. The development of wireless sensor networks was originally motivated by military applications, and today they are widely used in many industrial fields and everyday life, such as industrial process monitoring, traffic control, smart home, etc. If all sensors frequently transmit messages in a flooding way, then a lot of energy is wasted and intense interferences are created. To solve these problems, the concept of virtual backbone was proposed by Das and Bhargharan \cite{Das} and Ephremides {\it et al.} \cite{Ephremides} which corresponds to a connected dominating set in a graph.

Given a graph $G=(V,E)$, a subset $C$ of $V$ is said to be a {\em dominating set} $(DS)$ of $G$ if any $v\in V\setminus C$ is adjacent with at least one node of $C$. We say that a dominating set $C$ of $G$ is a {\em connected dominating set} of $G$ if $G[C]$ is connected, where $G[C]$ is the subgraph of $G$ induced by $C$. Nodes in $C$ are called {\em dominators}, while the other nodes are called {\em dominatees}.

In WSNs, a sensor may fail due to accidental damage or energy deletion. To make a virtual backbone more robust, it is suggested to use $(k,m)$-CDS.

\begin{definition}[$(k,m)$-CDS]
{\rm A node subset $C$ is a {\em $k$-connected $m$-fold dominating set}, if every node in $V\backslash C$ has at least $m$ neighbors in $C$ and $G[C]$ is $k$-connected.}
\end{definition}

In a homogeneous wireless sensor network, all sensors are equipped with omnidirectional antennas with the same transmission radius $($say, one unit$)$, and thus the transmission range of every sensor is a disk of radius one. Two sensors can communicate with each other if and only if they fall into the transmission ranges of each other. Such a setting is typically modeled as a {\em unit disk graph} $(UDG)$, in which every node of the graph corresponds to a sensor on the plane, and two nodes are adjacent if and only if the Euclidean distance between their corresponding sensors is at most one unit. There are a lot of studies on virtual backbones in UDG (see the book \cite{DuBookCDS}), but for general graphs, related studies are rare.

Notice that in a real world, the environment is very complicated, and thus it is rare that the topology can be ideally modeled as a unit disk graph. So, it is meaningful to study virtual backbone in a general graph.

In this paper, we study the minimum $(3,m)$-CDS problem with $m\geq3$ in a general graph. The strategy used in this paper is greedy. It is well known that if the potential function related with the greedy algorithm is monotone increasing and submodular, then an $O(\ln n)$ approximation ratio can be achieved. An interesting part of this paper is that we constructed a potential function which is NOT submodular, and proposed an analysis to show that the approximation ratio $O(\ln n)$ can still be achieved.

The main result of this paper is that our algorithm works for general graphs with a guaranteed performance ration $(\alpha+8+2\ln(2\alpha-6))$ for $\alpha\geq4$ and a guaranteed performance ration $(3\alpha+2\ln2)$ for $\alpha<4$, where $\alpha$ is the approximation ratio for the minimum $(2,m)$-CDS problem. In a recent paper, we \cite{Shi} proposed a $\big(\ln(\delta +m-2)+o(\ln \delta)\big)$-approximation algorithm for the minimum $(2,m)$-CDS problem on a general graph, where $\delta$ is the maximum degree of the graph. Based on it, the algorithm in this paper has performance ratio $\ln(\delta+m-2)+o(\ln \delta)$. In view of the non-approximability of this problem \cite{Guha}, the ratio is asymptotically best possible.

Furthermore, if applying our algorithm on a unit disk graph, then the performance ratio is less than $27$. Previous to this work, Wang {\it et al.} \cite{Wang1} obtained a constant approximation algorithm for $(3,m)$-CDS on UDG, and the ratio is further improved in their recent work \cite{Wang2}, which is $5\alpha$. For example, if the value of $\alpha$ in paper \cite{Shi} is used, their algorithm for $(3,3)$-CDS on UDG has performance ratio $62.3$. Our ratio improves theirs by a large amount.

Our work is based on the brick decomposition of 2-connected graphs, which is commonly known as Tutte's decomposition. This decomposition is an important tool in graph theory, and was studied extensively by a lot of researchers, including Tutte \cite{Tutte}, Hopcroft and Tarjan \cite{Hopcroft}, Cunningham and Edmonds \cite{Cunningham}, {\it et al.}. The same decomposition is also used by Wang {\it et al.} \cite{Wang2}. However, our method differs a lot from theirs since we are considering general graphs while they only considered unit disk graphs. Furthermore, our method is more refined which can be seen from the improvement on the performance ratio.

The rest of this paper is organized as follows. Section \ref{sec.2} introduces related works. Some preliminary results concerning with the brick decomposition structure of 2-connected graphs are introduced in Section \ref{sec3}. In Section \ref{sec14-9-7-1}, the algorithm is presented, and the performance ratio is analyzed. Section \ref{sec5} concludes the paper and discusses some future research directions.

\section{Related work}\label{sec.2}

The idea of using a CDS as a virtual backbone for WSN was proposed by Das and Bhargharan \cite{Das} and Ephremides {\it et al.} \cite{Ephremides}. The minimum CDS problem is NP-hard. In fact, by reducing the minimum set cover problem to the minimum CDS problem, Guha and Khuller \cite{Guha} proved that a minimum CDS cannot be approximated within $\rho\ln n$ for any $0<\rho<1$ unless $NP\subseteq DTIME(n^{O(loglog n)})$. In the same paper, they proposed two greedy algorithms with performance ratios of $2(H(\delta)+1)$ and $H(\delta)+2$, respectively, where $\delta$ is the maximum degree of the graph and $H(\cdot)$ is the harmonic number. This was improved by Ruan {\it et al.} \cite{Ruan} to $2+\ln\delta$. Du {\it et al.} \cite{Du} presented a $(1+\varepsilon)(1+\ln(\delta-1))$-approximation algorithm, where $\varepsilon$ is an arbitrary positive real number. In UDGs, a polynomial time approximation scheme $($PTAS$)$ for this problem was given by Cheng {\it et al.} \cite{Cheng}, which was generalized to higher dimensional space by Zhang {\it et al.} \cite{Zhang}. For distributed algorithms with constant performance ratios, the readers may refer to \cite{DuYingfan,Li1,LiYingshu,Wan,Wan1,Wu}.

The problem of constructing fault-tolerant virtual backbone was proposed by Dai and Wu \cite{Dai}. They proposed three heuristic algorithms for the minimum $(k,k)$-CDS problem. However, no theoretical analysis was given. Table \ref{tab1} summarizes results with guaranteed performance ratio for $(k,m)$-CDS. The last two rows are results obtained in this paper. It can be seen that we obtained the first approximation algorithm for $(3,m)$-CDS on a general graph. When the algorithm is applied on UDG, the performance ratio is reduced by a large amount compared with previous ones. For some heuristics on $(k,m)$--MCDS for general $k$ and $m$, the readers may refer to \cite{Li2,Thai,Wu1}.

\renewcommand{\arraystretch}{1.3}

\begin{table}[h!]
\centering
\begin{tabular}{ |c|c|c|c| }
\hline
graph & $(k,m)$ & ratio & reference \\
\hline
general & $(1,m)$ & $2H(\delta+m-1)$ & \cite{Li,Zhang1} \\
\hline
general & $(1,m)$ & $2+H(\delta+m-2)$ & \cite{Zhou} \\
\hline
general & $(2,m)$ & $4+\ln(\delta+m-2)+2\ln (2+\ln (\delta+m-2))$ & \cite{Shi} \\
\hline
UDG & $(2,1)$ & 72 & \cite{Wang} \\
\hline
UDG & $(1,m)$ & $\left\{\begin{array}{ll}5+5/m, & m\leq 5\\ 7, & m>5\end{array}\right.$ & \cite{Shang}\\
\hline
UDG & $(2,m)$ & $\left\{\begin{array}{ll}15+15/m, & 2\leq m\leq 5\\ 21, & m>5\end{array}\right.$ & \cite{Shang}\\
\hline
UDG & $(2,m)$ & $\left\{\begin{array}{l}7+5/m+2\ln (5+5/m),\ \mbox{for}\  2\leq m\leq 5\\ 12.89,\qquad\qquad\mbox{for}\  m>5\end{array}\right.$ & \cite{Shi}\\
\hline
UDG & $(3,m)$ & constant (280 for $m=3$) & \cite{Wang1}\\
\hline
UDG & $(3,m)$ & $\begin{array}{ll}\mbox{$5\alpha$, where $\alpha$ is performance ratio for $(2,m)$-CDS}\\ \mbox{on UDG (62.3 for $m=3$)}\end{array}$ & \cite{Wang2}\\
%\hline
%UDG & $(k,m)$ & $\left\{\begin{array}{l}(5+5/m)(\delta(\delta-1)^2)^{k-1}\ \mbox{for}\  m\leq 5\\ 7(\delta(\delta-1)^2)^{k-1},\ \mbox{for}\  m>5\end{array}\right.$  & \cite{Li2}\\
\hline
general & $(3,m)$ & $\begin{array}{c}
\left\{\begin{array}{l}\alpha+8+2\ln(2\alpha-6)\ \mbox{for}\  \alpha\geq4\\ 3\alpha+2\ln2,\ \mbox{for}\  \alpha<4\end{array}\right.\\ \mbox{where $\alpha$ is performance ratio for $(2,m)$-CDS} \end{array}$ & *\\

\hline
UDG & $(3,m)$ & $\left\{\begin{array}{ll} 26.34, & m=3\\ 25.68, & m=4\\26.86, & m\geq 5\end{array}\right.$ & *\\
\hline
\end{tabular}
\vskip 0.2cm \caption{Results on $(k,m)$-CDS with guaranteed performance ratio}
\label{tab1}
\end{table}

\section{Preliminaries}\label{sec3}

The following lemma is well known in graph theory \cite{Bondy}.

\begin{lemma}\label{lem14-6-9-1}
Suppose $H_1$ is a $k$-connected graph and
$H_2$ is obtained from $H_1$ by adding a new
node $u$ and joining $u$ to at least $k$ nodes
of $H_1$. Then $H_2$ is also $k$-connected.
\end{lemma}

As a consequence, we have the following result.
\begin{corollary}\label{cor14-5-9-2}
Suppose $G$ is a $k$-connected graph, $k$ and $m$ are two positive integers with
$m\geq k$, and $C$ is a $(k,m)$-CDS of $G$.
For any $U\subseteq V(G)\setminus C$,

$(\romannumeral1)$ node set $C\cup U$
is also a $(k,m)$-CDS of $G$, and

$(\romannumeral2)$ no node in $U$ is involved in any $k$-node cut of $G[C\cup U]$.
\end{corollary}

In the following, we focus on $2$-connected graphs.

\begin{definition}[2-separator]
{\rm Suppose $H$ is a $2$-connected graph. A node set $\{u,v\}$ is a {\em $2$-separator} of $H$ if $H-\{u,v\}$ is not connected. The {\em local connectivity} between two nodes $u$ and $v$ in graph $H$ is the maximum number of internally disjoint $(u,v)$-paths in $H$, denoted as $p_H(u,v)$. A 2-separator $\{u,v\}$ is {\em good} if $p_{H}(u,v)\geq3$, otherwise it is {\em bad}.}
\end{definition}

For example, in Fig.\ref{fig14-6-9-5}$(a)$, $\{u_1,v_1\}$ is a good 2-separator of the first graph and $\{u_2,v_2\}$ is a bad 2-separator of the graph containing it (which is a 4-cycle). The following lemma characterizes 2-connected graphs without good 2-separators.

\begin{lemma}[\cite{Zhang2}]\label{lemma2}
Let $H$ be a $2$-connected graph which has no good $2$-separator. Then $H$ is either $3$-connected or a cycle.
\end{lemma}

In view of Lemma \ref{lemma2}, we say that a $2$-connected graph without good $2$-separators is a {\em $T$-brick} if it is $3$-connected or an {\em $R$-brick} if it is a cycle.

Suppose $H$ is a $(2,m)$-CDS of a 3-connected graph $G$, where $m\geq 3$. In view of Corollary \ref{cor14-5-9-2}, adding nodes to $H$ does not incur new 2-separators. So, to augment $H$ into a $(3,m)$-CDS, it suffices to eliminate all 2-separators in $H$. However, the number of 2-separators might be exponential. In order that the algorithm is polynomial, 2-separators have to be eliminated in a neat way. For this purpose, we need a structural characterization of 2-connected graphs, based on the concept of marked components defines as follows.

\begin{definition}[$S$-component and marked $S$-component]
{\rm Let $H$ be a 2-connected graph, $S=\{u,v\}$ be a $2$-separator of $H$, and $C$ be a connected component of $H-S$. The subgraph $H[C\cup S]$ is called an {\em $S$-component} of $H$. For an $S$-component $H[C\cup S]$, add a virtual edge $uv$ if $uv\notin E(H)$ and do nothing if $uv\in E(H)$, call the resulting graph as a {\em marked $S$-component}.}
\end{definition}

For example, in the first graph of Fig.\ref{fig14-6-9-5}$(a)$, $S_1=\{u_1,v_1\}$ is a 2-separator. Splitting off the graph through $S_1$ results in three marked $S_1$-components as in the second graph of Fig.\ref{fig14-6-9-5}$(a)$. Those dotted edges are virtual edges. The role virtual edges play is to guarantee the 2-connectedness of marked components, as indicated by Lemma \ref{lemma1} whose proof can be found in \cite{Bondy}.

\begin{figure*}[!htbp]\setlength{\unitlength}{0.9pt}
\begin{center}
\begin{picture}(40,130)
\put(0,0){\circle*{5}}
\put(40,0){\circle*{5}}
\put(0,25){\circle*{5}}
\put(40,25){\circle*{5}}
\put(0,50){\circle*{5}}
\put(40,50){\circle*{5}}
\put(20,70){\circle*{5}}
\put(0,90){\circle*{5}}
\put(40,90){\circle*{5}}
\put(0,130){\circle*{5}}
\put(40,130){\circle*{5}}
\qbezier(0,0)(20,0)(40,0)
\qbezier(0,0)(20,12.5)(40,25)
\qbezier(0,0)(20,25)(40,50)
\qbezier(0,25)(20,12.5)(40,0)
\qbezier(0,25)(20,25)(40,25)
\qbezier(0,25)(20,37.5)(40,50)
\qbezier(0,50)(20,25)(40,0)
\qbezier(0,50)(20,37.5)(40,25)
\qbezier(0,50)(20,50)(40,50)
\qbezier(0,50)(0,90)(0,130)
\qbezier(40,50)(40,90)(40,130)
\qbezier(20,70)(10,80)(0,90)
\qbezier(20,70)(30,80)(40,90)
\qbezier(0,90)(20,110)(40,130)
\qbezier(40,90)(20,110)(0,130)
\qbezier(0,130)(20,130)(40,130)
\put(48,87){\vector(2,-1){30}}
\put(48,90){\vector(2,1){30}}
\put(48,93){\vector(2,3){30}}
\put(-14,88){$u_1$}\put(25,88){$v_1$}\put(-34,88){$S_1$}
\end{picture}
\hskip 1.3cm\begin{picture}(40,170)
\put(0,0){\circle*{5}}
\put(40,0){\circle*{5}}
\put(0,25){\circle*{5}}
\put(40,25){\circle*{5}}
\put(0,50){\circle*{5}}
\put(40,50){\circle*{5}}
\put(0,90){\circle*{5}}
\put(40,90){\circle*{5}}
\put(20,100){\circle*{5}}
\put(0,120){\circle*{5}}
\put(40,120){\circle*{5}}
\put(0,130){\circle*{5}}
\put(40,130){\circle*{5}}
\put(0,170){\circle*{5}}
\put(40,170){\circle*{5}}
\put(12,-20){(a)}
\put(42,148){$B_1$}\put(42,103){$B_2$}
\qbezier(0,0)(20,0)(40,0)
\qbezier(0,0)(20,12.5)(40,25)
\qbezier(0,0)(20,25)(40,50)
\qbezier(0,25)(20,12.5)(40,0)
\qbezier(0,25)(20,25)(40,25)
\qbezier(0,25)(20,37.5)(40,50)
\qbezier(0,50)(20,25)(40,0)
\qbezier(0,50)(20,37.5)(40,25)
\qbezier(0,50)(20,50)(40,50)
\qbezier(0,50)(0,70)(0,90)
\qbezier(40,50)(40,70)(40,90)
{\linethickness{0.3mm}\qbezier[17](0,90)(20,90)(40,90)}
\qbezier(20,100)(10,110)(0,120)
\qbezier(20,100)(30,110)(40,120)
{\linethickness{0.3mm}\qbezier[17](0,120)(20,120)(40,120)}
{\linethickness{0.3mm}\qbezier[17](0,130)(20,130)(40,130)}
\qbezier(0,130)(0,150)(0,170)
\qbezier(0,130)(20,150)(40,170)
\qbezier(40,130)(20,150)(0,170)
\qbezier(40,130)(40,150)(40,170)
\qbezier(0,170)(20,170)(40,170)
\put(48,52){\vector(3,2){30}}\put(48,48){\vector(2,-1){30}}
\put(-16,46){$S_2$}
\end{picture}
\hskip 1.3cm\begin{picture}(40,100)
\put(0,0){\circle*{5}}
\put(40,0){\circle*{5}}
\put(0,25){\circle*{5}}
\put(40,25){\circle*{5}}
\put(0,50){\circle*{5}}
\put(40,50){\circle*{5}}
\put(0,60){\circle*{5}}
\put(40,60){\circle*{5}}
\put(0,100){\circle*{5}}
\put(40,100){\circle*{5}}
\put(42,76){$B_3$}\put(42,22){$B_4$}
\qbezier(0,0)(20,0)(40,0)
\qbezier(0,0)(20,12.5)(40,25)
\qbezier(0,0)(20,25)(40,50)
\qbezier(0,25)(20,12.5)(40,0)
\qbezier(0,25)(20,25)(40,25)
\qbezier(0,25)(20,37.5)(40,50)
\qbezier(0,50)(20,25)(40,0)
\qbezier(0,50)(20,37.5)(40,25)
\qbezier(0,50)(20,50)(40,50)
\qbezier(0,60)(20,60)(40,60)
\qbezier(0,60)(0,80)(0,100)
\qbezier(40,60)(40,80)(40,100)
\put(3,64){$u_2$}\put(26,90){$v_2$}
{\linethickness{0.3mm}\qbezier[17](0,100)(20,100)(40,100)}
\end{picture}
\hskip 2.4cm\begin{picture}(40,130)
\put(0,0){\circle*{5}}
\put(40,0){\circle*{5}}
\put(0,25){\circle*{5}}
\put(40,25){\circle*{5}}
\put(0,50){\circle*{5}}
\put(40,50){\circle*{5}}
\put(-20,76){\circle*{5}}
\put(0,90){\circle*{5}}
\put(40,90){\circle*{5}}
\put(0,130){\circle*{5}}
\put(40,130){\circle*{5}}
\put(12,-20){(b)}
\qbezier(0,0)(20,0)(40,0)
\qbezier(0,0)(20,12.5)(40,25)
\qbezier(0,0)(20,25)(40,50)
\qbezier(0,25)(20,12.5)(40,0)
\qbezier(0,25)(20,25)(40,25)
\qbezier(0,25)(20,37.5)(40,50)
\qbezier(0,50)(20,25)(40,0)
\qbezier(0,50)(20,37.5)(40,25)
\qbezier(0,50)(20,50)(40,50)
\qbezier(0,50)(0,90)(0,130)
\qbezier(40,50)(40,90)(40,130)
\qbezier(-20,76)(-10,83)(0,90)
\qbezier(-20,76)(10,83)(40,90)
{\linethickness{0.3mm}\qbezier[17](0,90)(20,90)(40,90)}
\qbezier(0,90)(20,110)(40,130)
\qbezier(40,90)(20,110)(0,130)
\qbezier(0,130)(20,130)(40,130)
\put(20,25){\oval(60,62)}
\put(20,70){\oval(60,52)}
\put(20,110){\oval(60,52)}
\put(10,85){\oval(90,32)}
\end{picture}
\hskip 1.5cm\begin{picture}(40,130)
\put(0,90){\circle*{5}}
\put(30,0){\circle*{5}}
\put(30,25){\circle*{5}}
\put(30,50){\circle*{5}}
\put(30,90){\circle*{5}}
\put(30,130){\circle*{5}}
\put(15,-20){(c)}
\qbezier(0,90)(15,90)(30,90)
\qbezier(30,0)(30,65)(30,130)
\put(35,128){$B_1$}\put(-2,94){$B_2$}\put(35,48){$B_3$}\put(35,-2){$B_4$}
\put(35,88){$S_1$}\put(35,23){$S_2$}
\end{picture}
\vskip 0.5cm\caption{$(a)$ Decomposition through good separators. Dotted edges are virtual edges which are added to form the marked $S$-components. $(b)$ Brick structure of graph $H$. Each ellipse indicates a brick. $(c)$ The brick-tree $B(H)$.}\label{fig14-6-9-5}
\end{center}
\end{figure*}
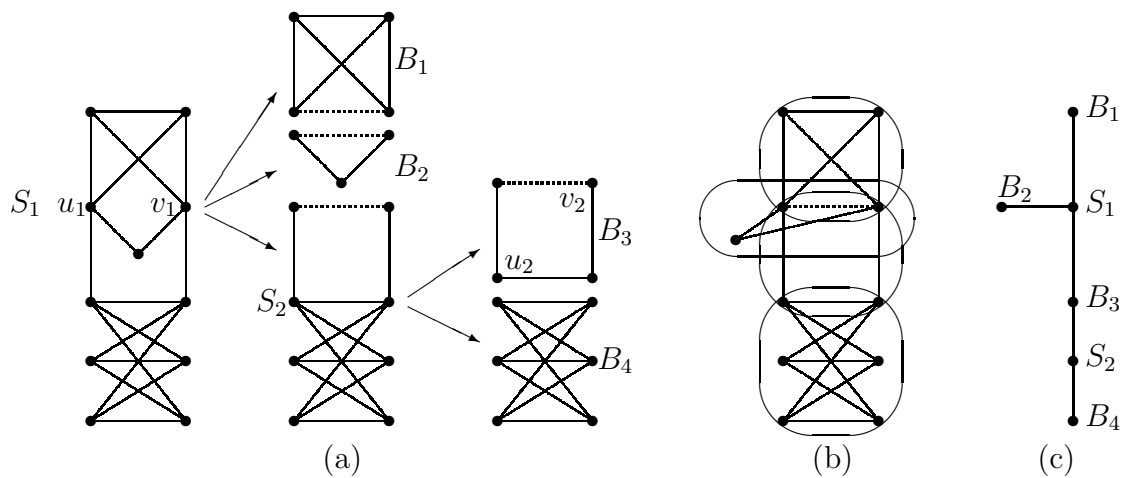

\begin{lemma}\label{lemma1}
Let $H$ be a 2-connected graph and $S$ be a $2$-separator of $H$. Then the marked $S$-components of $H$ are also 2-connected.
\end{lemma}

Let $G$ be a 3-connected graph and $H$ be the subgraph of $G$ induced by a $(2,m)$-CDS of $G$. If $H$ has a good $2$-separator $S$, then it can be decomposed into several marked $S$-components, which are also 2-connected by Lemma \ref{lemma1}. If any one of these marked $S$-components has a good $2$-separator, it can be further decomposed into smaller marked components. Such a decomposition continues until $H$ is decomposed into marked components without good 2-separators. In other words, $H$ can be iteratively decomposed into $T$-bricks and $R$-bricks through good 2-separators.

Pasting these bricks through those good $2$-separators which have been used in the decomposition procedure, we see that the brick structure of $H$ is tree-like in the following sense: Let $B(H)$ be a bipartite graph with bipartition $(\mathcal B,\mathcal S)$, where $\mathcal B$ is the set of bricks and $\mathcal S$ is the set of good 2-separators used in the above decomposition. A brick $B\in \mathcal B$ is adjacent with a 2-separator $S\in\mathcal S$ if and only if $S$ is contained in $B$. Notice that there is no sequence of bricks $B_1,\ldots,B_t$ such that $B_i$ shares a 2-separator with $B_{i+1}$ for $i=1,\ldots,t-1$ and $B_t$ shares a 2-separator with $B_1$ (otherwise $\bigcup_{i=1}^tB_i$ will be 3-connected). So, the graph $B(H)$ is acyclic. Clearly, $B(H)$ is connected. So, $B(H)$ is a tree, which is called the {\em brick-tree} of $H$.
Such a decomposition is illustrated in Fig.\ref{fig14-6-9-5}.

\section{Algorithm and Analysis}\label{sec14-9-7-1}

This section presents our greedy algorithm and analyzes its performance ratio. We first construct a potential function $f$ which will be used in the greedy algorithm, and derive some properties about $f$.

\subsection{Potential Function}\label{subsecFunction}

\begin{definition}[brick-bridge]\label{def6-1-1}
{\rm Suppose $H$ is a $2$-connected graph. A path $P$ is called a
{\em brick-bridge of $H$} if it satisfies the following three conditions:

$(\romannumeral1)$ all internal nodes of $P$ are outside of $H$ and the two ends of $P$ are in $H$;

$(\romannumeral2)$ the two ends of $P$ are nonadjacent in $H$;

$(\romannumeral3)$ the two ends of $P$ do not belong to a same
$T$-brick of $H$.

\noindent Denote by $int(P)$ the set of internal nodes of $P$.}
\end{definition}

By the above definition, any brick-bridge either ``strides over'' different bricks or ``strides over'' non-adjacent nodes of an $R$-brick.

As we have explained in Section \ref{sec3}, the assumption $m\geq 3$ guarantees that adding brick-bridges to a $(2,m)$-CDS does not incur new 2-separators. Fig.\ref{fig14-6-18-2} gives us some idea of how the brick-structure is changed after adding internal nodes of some brick-bridge. Roughly speaking, if the brick-bridge $P$ strides over bricks $B$ and $B'$ of $G[C]$, let $Q_{BB'}$ be the unique path on the brick tree of $G[C]$ connecting $B$ and $B'$, and let $\mathcal Q_{BB'}$ be the set of bricks on $Q_{BB'}$, then all $T$-bricks in $\mathcal Q_{BB'}$ are merged into a new $T$-brick of $G[C\cup int(P)]$, and every $R$-brick in $\mathcal Q_{BB'}$ is divided into smaller $R$-bricks of $G[C\cup int(P)]$ by this new $T$-brick.

\begin{figure*}[!htbp]
\begin{center}
\hskip -3cm\begin{picture}(160,110)\setlength{\unitlength}{0.9pt}
\put(20,60){\circle*{5}}\put(40,60){\circle*{5}}
\put(120,60){\circle*{5}}\put(140,60){\circle*{5}}\put(160,60){\circle*{5}}\put(80,60){\circle*{5}}
\put(50,100){\circle*{5}}\put(110,100){\circle*{5}}
\put(55,74){\circle*{5}}\put(65,74){\circle*{5}}%\put(100,75){\circle*{5}}
\put(40,30){\circle*{5}}\put(80,30){\circle*{5}}\put(120,30){\circle*{5}}
\put(60,15){\circle*{5}}\put(100,15){\circle*{5}}
\qbezier(0,30)(0,45)(0,60)\qbezier(0,30)(20,30)(40,30)\qbezier(0,60)(20,60)(40,60)\qbezier(40,30)(40,45)(40,60)
\qbezier(80,30)(80,45)(80,60)\qbezier(40,30)(60,0)(80,30)\qbezier(80,30)(100,0)(120,30)\qbezier(40,60)(60,90)(80,60)\qbezier(80,60)(100,90)(120,60)
\qbezier(120,30)(120,45)(120,60)\qbezier(120,30)(140,30)(160,30)\qbezier(120,60)(140,60)(160,60)\qbezier(160,30)(160,45)(160,60)
{\linethickness{0.5mm}\qbezier[9](40,60)(45,80)(50,100)\qbezier[9](140,60)(125,80)(110,100)}
{\linethickness{0.5mm}\qbezier[9](20,60)(35,80)(50,100)\qbezier[11](50,100)(80,100)(110,100)\qbezier[11](160,60)(135,80)(110,100) }
\put(70,-5){(aa)}
\end{picture}
\hskip 1cm\begin{picture}(60,110)\setlength{\unitlength}{0.9pt}
\put(20,60){\circle*{5}}\put(40,60){\circle*{5}}\put(80,60){\circle*{5}}
\put(45,65){\circle*{5}}\put(75,65){\circle*{5}}
\put(120,60){\circle*{5}}\put(140,60){\circle*{5}}\put(160,60){\circle*{5}}
\put(50,100){\circle*{5}}\put(110,100){\circle*{5}}
\put(55,79){\circle*{5}}\put(65,79){\circle*{5}}%\put(100,75){\circle*{5}}
\put(40,30){\circle*{5}}\put(80,30){\circle*{5}}\put(120,30){\circle*{5}}
\put(45,25){\circle*{5}}\put(75,25){\circle*{5}}\put(85,25){\circle*{5}}\put(115,25){\circle*{5}}
\put(60,10){\circle*{5}}\put(100,10){\circle*{5}}
\qbezier(0,30)(0,45)(0,60)\qbezier(0,30)(20,30)(40,30)\qbezier(0,60)(20,60)(40,60)\qbezier(40,30)(40,45)(40,60)\qbezier(80,30)(80,45)(80,60)
\qbezier(45,25)(60,-5)(75,25)\qbezier(85,25)(100,-5)(115,25)\qbezier(45,65)(60,95)(75,65)\qbezier(80,60)(100,90)(120,60)
\qbezier(120,30)(120,45)(120,60)\qbezier(120,30)(140,30)(160,30)\qbezier(120,60)(140,60)(160,60)\qbezier(160,30)(160,45)(160,60)
\qbezier(20,60)(35,80)(50,100)\qbezier(40,60)(45,80)(50,100)\qbezier(50,100)(80,100)(110,100)\qbezier(140,60)(125,80)(110,100)\qbezier(160,60)(135,80)(110,100)
{\linethickness{0.35mm}\qbezier[13](45,65)(60,65)(75,65)\qbezier[17](40,60)(60,60)(80,60)
\qbezier[35](40,30)(80,30)(120,30)\qbezier[13](45,25)(60,25)(75,25)\qbezier[13](85,25)(100,25)(115,25)}
\put(70,-5){(ab)}
\end{picture}

\hskip -1cm\begin{picture}(80,110)
\put(20,60){\circle*{5}}\put(40,60){\circle*{5}}\put(80,60){\circle*{5}}
\put(50,100){\circle*{5}}
\put(55,74){\circle*{5}}\put(65,74){\circle*{5}}
\put(40,30){\circle*{5}}\put(80,30){\circle*{5}}
\put(60,15){\circle*{5}}
\qbezier(0,30)(0,45)(0,60)\qbezier(0,30)(20,30)(40,30)\qbezier(0,60)(20,60)(40,60)\qbezier(40,30)(40,45)(40,60)
\qbezier(80,30)(80,45)(80,60)\qbezier(40,30)(60,0)(80,30)\qbezier(40,60)(60,90)(80,60)
{\linethickness{0.5mm}\qbezier[9](20,60)(35,80)(50,100)\qbezier[8](40,60)(45,80)(50,100)\qbezier[7](50,100)(57.5,87)(65,74) }
\put(30,-5){(ba)}
\end{picture}
\hskip 0.3cm\begin{picture}(60,110)
\put(20,60){\circle*{5}}\put(40,60){\circle*{5}}\put(80,60){\circle*{5}}
\put(50,100){\circle*{5}}
\put(48,70){\circle*{5}}\put(54,79){\circle*{5}}\put(62,76){\circle*{5}}
\put(70,72){\circle*{5}}\put(76,68){\circle*{5}}
\put(40,30){\circle*{5}}\put(44,22){\circle*{5}}\put(80,30){\circle*{5}}
\put(64,14){\circle*{5}}
\qbezier(0,30)(0,45)(0,60)\qbezier(0,30)(20,30)(40,30)\qbezier(0,60)(20,60)(40,60)\qbezier(40,30)(40,45)(40,60)
\qbezier(80,30)(80,45)(80,60)\qbezier(44,22)(64,0)(80,30)\qbezier(48,70)(51,85)(62,76)\qbezier(76,68)(77,70)(80,60)
\qbezier(20,60)(35,80)(50,100)\qbezier(40,60)(45,80)(50,100)\qbezier(50,100)(60,86)(70,72)
{\linethickness{0.35mm} \qbezier[8](48,70)(55,73)(62,76)\qbezier[15](40,60)(55,66)(70,72)
\qbezier[25](40,30)(55,51)(70,72)\qbezier[25](44,22)(60,45)(76,68) }
\put(30,-5){(bb)}
\end{picture}
\hskip 2cm\begin{picture}(80,90)
\put(40,60){\circle*{5}}\put(80,60){\circle*{5}}
\put(55,74){\circle*{5}}\put(65,74){\circle*{5}}
\put(40,30){\circle*{5}}\put(80,30){\circle*{5}}
\put(60,15){\circle*{5}}\put(60,45){\circle*{5}}
\qbezier(0,30)(0,45)(0,60)\qbezier(0,30)(20,30)(40,30)\qbezier(0,60)(20,60)(40,60)\qbezier(40,30)(40,45)(40,60)
\qbezier(80,30)(80,45)(80,60)\qbezier(40,30)(60,0)(80,30)\qbezier(40,60)(60,90)(80,60)
{\linethickness{0.5mm}\qbezier[6](60,45)(57.5,59.5)(55,74)\qbezier[6](60,45)(62.5,59.5)(65,74)\qbezier[6](60,45)(60,30)(60,15) }
\put(30,-5){(ca)}
\end{picture}
\hskip 0.3cm\begin{picture}(60,90)
\put(40,60){\circle*{5}}\put(40,30){\circle*{5}}
\put(50,60){\circle*{5}}\put(50,30){\circle*{5}}
\put(100,60){\circle*{5}}\put(100,30){\circle*{5}}
\put(64,16){\circle*{5}}\put(74,15){\circle*{5}}\put(84,16){\circle*{5}}
\put(60,74){\circle*{5}}\put(68,75){\circle*{5}}\put(80,75){\circle*{5}}\put(88,74){\circle*{5}}
\put(74,45){\circle*{5}}
\qbezier(0,30)(0,45)(0,60)\qbezier(0,30)(20,30)(40,30)\qbezier(0,60)(20,60)(40,60)\qbezier(40,30)(40,45)(40,60)
\qbezier(50,30)(50,45)(50,60)\qbezier(100,30)(100,45)(100,60)
\qbezier(74,45)(74,30)(74,15)\qbezier(74,45)(71,60)(68,75)\qbezier(74,45)(77,60)(80,75)
\qbezier(50,60)(52,67)(60,74)\qbezier(50,30)(54,23)(64,16)
\qbezier(100,60)(98,70)(86,75)\qbezier(100,30)(95,20)(88,16)\qbezier(68,75)(74,77)(80,75)
{\linethickness{0.35mm} \qbezier[27](60,74)(62,45)(64,16)\qbezier[27](67,75)(70,45)(73,15)
\qbezier[27](81,75)(78,45)(75,15)\qbezier[27](88,74)(86,45)(84,16)  }
\put(30,-5){(cb)}
\end{picture}

\caption{The change of brick structure after the internal nodes of brick-bridge $P$ is added. In $(*a)$, each rectangle represents a $T$-brick of $G[C]$ and each rounded rectangle represents an $R$-brick of $G[C]$. The dotted edges in $(*a)$ are the edges added together with the addition of $int(P)$. Figures in $(*b)$ show the brick-decompositions of $G[C\cup int(P)]$. The dotted edges in $(*b)$ are the virtual edges which are used to create marked components.}\label{fig14-6-18-2}
\end{center}
\end{figure*}
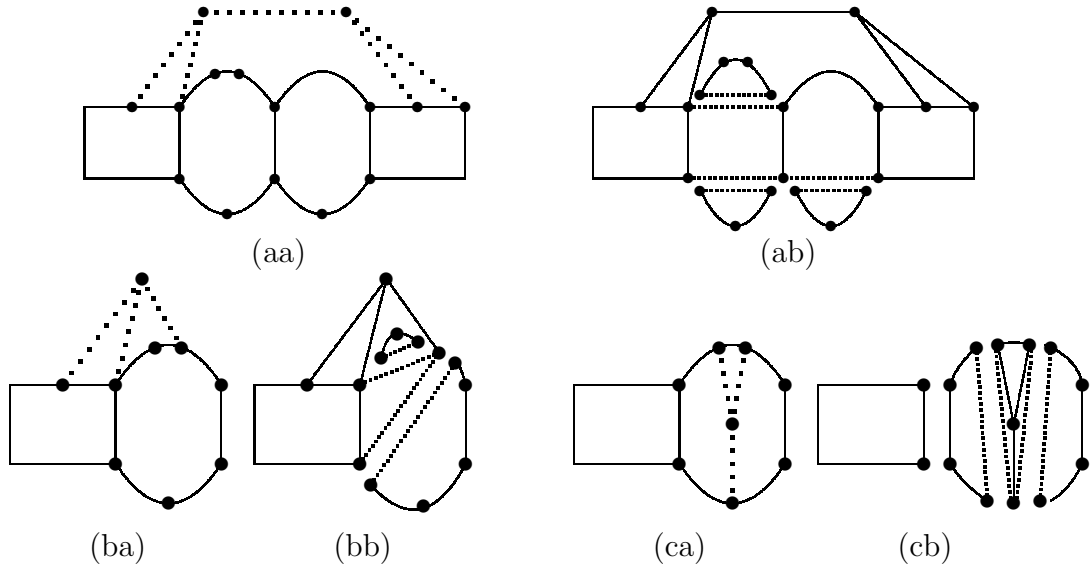

However, this rough description is not accurate. In fact, since we are considering {\em node-induced} subgraph, when the internal nodes of some brick-bridge is added, we are actually adding a lot of brick-bridges. Consider  Fig.\ref{fig14-6-9-3} for an example, $P=u_0u_1u_2u_3$ is a brick-bridge. Adding $int(P)=\{u_1,u_2\}$, another block-bridge $u_0u_1u_4$ is added as a byproduct. It should also be noted that we regard brick-bridge $u_0u_1u_4$ to stride over $B_1$ and $B_3$, not $B_1$ and $B_4$, since brick $B_4$ is not affected by adding brick-bridge $u_0u_1u_4$. The following observation is a more accurate description on the change of brick structure.

\begin{figure*}[!htbp]
\begin{center}
\begin{picture}(120,100)

\put(0,30){\circle*{5}}
\put(40,10){\circle*{5}}\put(40,30){\circle*{5}}\put(80,10){\circle*{5}}\put(80,30){\circle*{5}}
\put(50,100){\circle*{5}}\put(50,70){\circle*{5}}\put(9,66){\circle*{5}}

\qbezier(0,10)(-10,20)(0,30)
\qbezier(40,10)(40,20)(40,30)
\qbezier(80,10)(80,20)(80,30)
\qbezier(120,10)(130,20)(120,30)
\qbezier(0,10)(20,0)(40,10)
\qbezier(40,10)(60,0)(80,10)
\qbezier(80,10)(100,0)(120,10)
\qbezier(0,30)(20,40)(40,30)
\qbezier(40,30)(60,40)(80,30)
\qbezier(80,30)(100,40)(120,30)
\qbezier(40,30)(20,80)(50,80)
\qbezier(40,10)(80,70)(50,80)

{\linethickness{0.35mm}\qbezier[23](50,100)(0,77)(0,30)
\qbezier[19](50,100)(90,77)(80,30)
\qbezier[9](50,100)(50,85)(50,70) }

\put(46,61){$u_0$}\put(55,100){$u_1$}\put(-4,70){$u_2$}\put(-14,33){$u_3$}\put(83,25){$u_4$}
\put(42,45){$B_1$}\put(15,15){$B_2$}\put(60,15){$B_3$}\put(100,15){$B_4$}
\end{picture}
\vskip 0.3cm \caption{Adding the internal nodes of a brick-bridge may result in the addition of more than one brick-bridges.}\label{fig14-6-9-3}
\end{center}
\end{figure*}
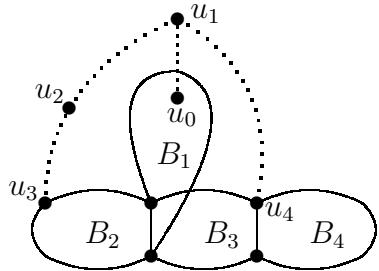

%%%%%

\begin{observation}\label{obs1}
Suppose $m\geq 3$, $G$ is a $3$-connected graph, and $C$ is a $(2,m)$-CDS of $G$. Let $X$ be a node set of $V(G)\setminus C$ such that $G[X]$ is connected. Denote $\mathcal B_X=\{(B,B')\colon$ $B,B'$ are bricks in $G[C]$ and there exists a brick-bridge of $G[C]$ whose internal nodes are in $X$ which strides over $B$ and $B'\}$. Let $\mathcal Q_X=\bigcup_{(B,B')\in\mathcal B_X}\mathcal Q_{BB'}$.

$(\romannumeral1)$ Those $T$-bricks of $\mathcal Q_X$ are merged into a bigger new $T$-brick of $G[C\cup X]$, and $X$ is contained in this new $T$-brick.

$(\romannumeral2)$ Each $R$-brick of $\mathcal Q_X$ is divided by the new $T$-brick into some smaller $R$-bricks of $G[C\cup X]$.

$(\romannumeral3)$ If an $R$-brick $R$ of $G[C]$ is divided into $s$ smaller $R$-bricks of $G[C\cup X]$, say $R_1,\ldots,R_s$, then
$$
\sum_{i=1}^s|R_i|\leq\left\{\begin{array}{ll}|R|+s, & \mbox{if}\ s\geq 3,\\|R|+1, & \mbox{if}\ s=2,\\ |R|-1, & \mbox{if}\ s=1.\end{array}\right.
$$
where $|R|$ is the number of nodes in $R$.

$(\romannumeral4)$ For every pair of bricks $(B,B')\in\mathcal B_X$, all those good 2-separators on the unique path $Q_{BB'}$ in the brick tree of $G[C]$ are contained in the new $T$-brick of $G[C\cup X]$.
\end{observation}

For a 2-connected graph $H$, denote by $\mathcal B(H)$ the set of bricks of $H$, $\mathcal R(H)$ the set of $R$-bricks of $H$, and $\mathcal T(H)$ the set of $T$-bricks of $H$. Define
$$
f(H)=|\mathcal T(H)|+q(H),
$$
where $q(H)=\Sigma_{R\in \mathcal R(H)}(2|R|-5)$.

\begin{lemma}\label{lemma11}
Suppose $m\geq 3$, $G$ is a $3$-connected graph, and $C$ is a $(2,m)$-CDS of $G$ such that $G[C]$ is not 3-connected. Let $P$ be a brick-bridge of $G[C]$ and let $X=int(P)$. Then $f(C)\geq f(C\cup X)+1$. If furthermore, $|\mathcal Q_X|\geq 2$ and there exists an $R$-brick $R_a\in\mathcal Q_X$ such that $|R_a|\geq4$, then $f(C)\geq f(C\cup X)+2$.
\end{lemma}
\begin{proof}
For each $R\in\mathcal R(G[C])$, we use $\mathcal R^{div}_{C,X}(R)$ to denote the set of smaller $R$-bricks arising from the division of $R$ after $X$ is added into $C$, and denote $s(R)=|\mathcal R^{div}_{C,X}(R)|$. For an integer $j\geq 0$, denote by $\mathcal R_j(C)$ (resp. $\mathcal R_{\geq j}(C)$) the set of $R$-bricks of $G[C]$ with $s(R)=j$ (resp. $s(R)\geq j$). Notice that every $R\in\mathcal R_0(C)$ is completely merged into the new $T$-brick and thus diminished from $\mathcal R(G[C\cup X])$. Let $\mathcal R^{rec}_X(C)$ be the set of $R$-bricks of $G[C]$ which remain the same in $G[C\cup X]$. For simplicity of notation, we use $\mathcal T(C)$ to denote $\mathcal T(G[C])$ etc. For any $R\in\mathcal R(G[C])$, observe that $|R|\geq 3$. Combining this with $(\romannumeral3)$ of Observation \ref{obs1}, we have
\begin{equation}\label{eq14-6-18-4}
\sum_{{R'\in \mathcal R^{div}_{C,X}(R)}}(2|R'|-5)\leq \left\{\begin{array}{ll}2|R|-3s(R), & \mbox{if}\ s(R)\geq 3,\\2|R|-8, & \mbox{if}\ s(R)=2,\\2|R|-7, & \mbox{if}\ s(R)=1,\\2|R|-6, & \mbox{if}\ s(R)=0.\end{array}\right.
\end{equation}
Then,
\begin{align*}
q(C\cup X) & =\sum_{R\in\mathcal R^{rec}_X(C)}(2|R|-5)+\sum_{R\in\mathcal R_{\geq 1}(C)}\sum_{R'\in\mathcal R^{div}_{C,X}(R)}(2|R'|-5)\\
& \leq \sum_{R\in\mathcal R^{rec}_X(C)}(2|R|-5)+\sum_{R\in\mathcal R_1(C)}(2|R|-7)+\!\!\!\! \sum_{R\in\mathcal R_2(C)}\!\!\!\!(2|R|-8)+\!\!\!\!\sum_{R\in\mathcal R_{\geq 3}(C)}\!\!\!\!(2|R|-3s(R))\nonumber\\
& =\sum_{R\in\mathcal R(C)\setminus\mathcal R_0(C)}(2|R|-5)-2|\mathcal R_1(C)|-3|\mathcal R_2(C)|-\sum_{R\in\mathcal R_{\geq 3}(C)}(3s(R)-5)\\
& \leq \sum_{R\in\mathcal R(C)}(2|R|-5)-\sum_{R\in\mathcal R_0(C)}(2|R|-5)-2|\mathcal R_{\geq1}(C)|.
\end{align*}
Hence
\begin{align}\label{eq14-6-23-11}
\Delta_Xq(C)=q(C\cup X)-q(C)\leq -\sum_{R\in\mathcal R_0(C)}(2|R|-5)-2|\mathcal R_{\geq1}(C)|.
\end{align}
By Observation \ref{obs1},
\begin{equation}\label{eq14-6-21-1}
\Delta_X|\mathcal T(C)|=|\mathcal T(C\cup X)|-|\mathcal T(C)|=1-|\mathcal Q_X^T(C)|,
\end{equation}
where $\mathcal Q_X^T(C)=\mathcal Q_X\cap \mathcal T(C)$ is the set of $T$-bricks of $G[C]$
which are merged into the new $T$-brick of $G[C\cup X]$. So,
\begin{align*}
\Delta_Xf(C)=\Delta_X|\mathcal T(C)|+\Delta_Xq(C)\leq 1-|\mathcal Q_X^T(C)|-\sum_{R\in\mathcal R_0(C)}(2|R|-5)-2|\mathcal R_{\geq 1}(C)|.
\end{align*}

If the lemma is not true, then $\Delta_Xf(C)=f(C\cup X)-f(C)\geq 0$, and thus
\begin{equation}\label{eq14-6-21-3}
|\mathcal Q_X^T(C)|+\sum_{R\in\mathcal R_0(C)}(2|R|-5)+2|\mathcal R_{\geq 1}(C)|\leq 1.
\end{equation}
It follows that $\mathcal R_{\geq 1}(C)=\emptyset$, $|\mathcal Q_X^T(C)|\leq 1$, and $|\mathcal R_0(C)|\leq 1$ (since every $R$-brick $R$ has at least three nodes, $2|R|-5\geq 1$). If $|\mathcal Q_X^T(C)|=1$, then by the definition of brick-bridge (the two ends of a brick-bridge do not belong to a same $T$-brick), we have $|\mathcal Q_X|\geq 2$, and thus $\mathcal Q_X$ has at least one $R$-brick. Since $\mathcal R_{\geq 1}(C)=\emptyset$, this $R$-brick belongs to $\mathcal R_0(C)$. But then $\sum_{R\in\mathcal R_0(C)}(2|R|-5)\geq 1$, and the left side of \eqref{eq14-6-21-3} is at least 2. So, all bricks of $\mathcal Q_X$ are $R$-bricks, and similarly to the above, they belong to $\mathcal R_0(C)$. Since $|\mathcal R_0(C)|\leq 1$, this is possible only when the brick-bridge $P$ strides over non-adjacent nodes of an $R$-brick $R$. It follows that $|R|\geq 4$, and thus $\sum_{R\in\mathcal R_0(C)}(2|R|-5)\geq 3$, again a contradiction. So, $f(C)\geq f(C\cup X)+1$. The first part of the lemma is proved.

Suppose the conditions for the second part of the lemma are satisfied. If $f(C)<f(C\cup X)+2$, then inequality \eqref{eq14-6-21-3} becomes
\begin{equation}\label{eq15-1-8-1}
|\mathcal Q_X^T(C)|+\sum_{R\in\mathcal R_0(C)}(2|R|-5)+2|\mathcal R_{\geq 1}(C)|\leq 2.
\end{equation}
We can not have $R_a\in\mathcal R_0(C)$, since otherwise the second term is at least $3$. Hence inequality \eqref{eq15-1-8-1} is possible only when $R_a\in\mathcal R_{\geq 1}(C)$, $|\mathcal R_{\geq 1}(C)|=1$, $\mathcal R_0(C)=\emptyset$, and $|\mathcal Q^T_X|=0$. But then, $|\mathcal Q_X|=1$, contradicting the assumption that $|\mathcal Q_X|\geq 2$. The second part of the lemma is proved.
\end{proof}

Lemma \ref{lemma11} says that as long as $G[C]$ is not 3-connected, the function $f$ can always be strictly decreased. Furthermore, under the ``if'' condition of Lemma \ref{lemma11}, the amount for the decrease can be at least 2.

\begin{lemma}\label{lem14-6-19-6}
Let $C$ be a node subset of $G$ such that $G[C]$ is 2-connected. Then, $f(C)=1$ if and only if either $G[C]$ is 3-connected or $G[C]$ is a triangle.
\end{lemma}
\begin{proof}
Notice that $f$ can also be written as
$$
f(C)=|\mathcal B(C)|+\sum_{R\in\mathcal R(C)}(2|R|-6).
$$
Since every $R$-brick has $|R|\geq 3$ and $|\mathcal B(C)|\geq 1$, we see that $f(C)=1$ if and only if $|\mathcal B(C)|=1$ and $\sum_{R\in\mathcal R(C)}(2|R|-6)=0$. Notice that $\sum_{R\in\mathcal R(C)}(2|R|-6)=0$ if and only if either $\mathcal R(C)=\emptyset$ or every $R\in\mathcal R(C)$ has $|R|=3$. In the first case, the unique brick of $G[C]$ is a $T$-brick, and thus $G[C]$ is 3-connected. In the second case, the unique brick of $G[C]$ is a cycle on three nodes, and thus a triangle.
\end{proof}

\begin{lemma}\label{lemma5}
Suppose $m\geq 3$, graph $G$ is 3-connected, and $C$ is a $(2,m)$-CDS of $G$ with $f(C)>1$. Then, there exists a brick-bridge $P$ of $G[C]$ with at most two internal nodes. Furthermore, if $G[C]$ is not a cycle, then for any brick $B\in \mathcal B(C)$, there exists such a brick-bridge $P$ of $G[C]$ satisfying $|\mathcal Q_{int(P)}|\geq2$ and $B\in\mathcal Q_{int(P)}$.
\end{lemma}

\begin{proof}
Since $f(C)>1$, by Lemma \ref{lem14-6-19-6}, $G[C]$ is not 3-connected. Let $S$ be a 2-separator of $G[C]$, and $G_1$ be a connected component of $G[C]-S$, $G_2$ be the union of the remaining connected components of $G[C]-S$. Since $G$ is 3-connected, there is a shortest path $P=u_{0}u_{1}\ldots u_{t}$ in $G$ between $G_1$ and $G_2$. Suppose $u_{0}\in V(G_{1})$ and $u_{t}\in V(G_{2})$. Assume $t\geq 4$. Since $C$ is an $m$-fold dominating set with $m\geq 3$, we see that $u_{2}$ has at least three neighbors in $C$, one of which is $v\notin S$. If $v\in V(G_1)$, then $vu_{2}\ldots u_{t}$ is a shorter path between $G_{1}$ and $G_{2}$. If $v\in V(G_2)$, then $u_{0}u_{1}u_{2}v$ is a shorter path between $G_1$ and $G_2$. Both cases contradict the shortest assumption on $P$. So, $t\leq 3$ and thus $|int(P)|\leq 2$.

Under the assumption that $G[C]$ is not a cycle and $f(C)>1$ (which implies that $G[C]$ is not 3-connected), we see from Lemma \ref{lemma2} that any brick $B\in\mathcal B(C)$ contains a good 2-separator. Use this good 2-separator as $S$ in the above proof. If $B$ is a $T$-brick, then $B-S$ is connected. If $B$ is an $R$-brick, then $S$ consists of two consecutive nodes on cycle $B$, and thus $B-S$ is also connected. So, we can take the connected component $G_1$ of $G[C]-S$ in the above proof such that $B-S\subseteq G_1$. Then it can be seen that the brick-bridge $P$ found by the above proof satisfies $|\mathcal Q_{int(P)}|\geq2$ and $B\in\mathcal Q_{int(P)}$.
\end{proof}

\subsection{Algorithm}\label{subsec4}

Our greedy algorithm is described in Algorithm \ref{algo1} with potential function $f(C)$. Initially, it computes a $(2,m)$-CDS $C_0$ by an existing algorithm, for example the one in \cite{Zhou}. If $G[C_0]$ is a triangle, then every node in $V(G)\setminus C_0$ is adjacent with all the three nodes of $C_0$ because $m\geq 3$. Hence, adding any node into $C_0$ results in a $K_4$ (complete graph on four nodes) which is a $(3,m)$-CDS of $G$. Suppose $G[C_0]$ is not a triangle. By Lemma \ref{lemma5}, as long as $f(C)>1$, there exists a brick-bridge $P$ with at most two internal nodes. By Lemma \ref{lemma11}, adding $int(P)$ strictly decreases the $f$-value. The while-loop iterates until $f(C)$ is decreased to 1, at which time $G[C]$ is 3-connected by Lemma \ref{lem14-6-19-6}.

\begin{algorithm}[h!]
\caption{\textbf{Computation of $(3,m)$-CDS for $m\geq 3$}} \label{algo1} Input: A $3$-connected
 graph $G=(V,E)$.

Output: A $(3,m)$-CDS $C$ of $G$.
\begin{algorithmic}[1]
    \State Compute a $(2,m)$-CDS $C_{0}$ by an $\alpha$-approximation algorithm.
    \If{$G[C_0]$ is a triangle}
        \State Let $v$ be an arbitrary node in $V(G)\setminus C_0$.
        \State Output $C$ $\leftarrow$ $C_0\cup\{v\}$.
    \Else
        \State $C$ $\leftarrow$ $C_{0}$.
        \While{$f(C)>1$}
            \State Select a brick-bridge $P$ of $G[C]$ with internal node set $int(P)=X$ such that $|X|\leq 2$ and $\frac{-\bigtriangleup_{X} f(C)}{|X|}$ is maximized.
            \State $C \leftarrow C\cup \{X\}$
        \EndWhile
        \State Output $C$.
    \EndIf
\end{algorithmic}\label{algo14-6-20-1}
\end{algorithm}

\subsection{Analysis of Performance Ratio}

To analyze the performance ratio of Algorithm \ref{algo14-6-20-1}, we first present a decomposition result on an optimal solution.

\begin{lemma}\label{lemma6}
Suppose $m\geq 3$, $C$ is a $(2,m)$-CDS of $G$, and $C^{\ast}$ is a minimum $(3,m)$-CDS of $G$. Then $C^{\ast}\backslash C$ can be decomposed into the union of node sets $C^{\ast}\backslash C=Y_{1}\cup Y_{2}\cup\ldots\cup Y_{h}$ satisfying the following conditions. For $j=1,2,\ldots,h,$ denote $C_{j}^{\ast}=Y_{1}\cup\ldots\cup Y_{j}$, $C_{0}^{\ast}=\emptyset$. Suppose $l$ is the first index such that $G[C\cup C_{l}^{\ast}]$ is 3-connected.

$(\romannumeral1)$  For $1\leq j\leq l,$ node set $C_{j}^{\ast}$ is completely contained in one $T$-brick of $G[C\cup C_{j}^{\ast}]$. Denote this brick as $B^{(j)}$, set $B^{(0)}=\emptyset$.

$(\romannumeral2)$  For $1\leq j\leq l,$ $Y_{j}=int(P_{j})$, where $P_{j}$ is a brick-bridge of $G[C]$ and there exists at least one brick of $G[C]$ contained in $B^{(j-1)}$ which also belongs to $\mathcal Q_{Y_j}(C)$.

$(\romannumeral3)$  $1\leq |Y_{j}|\leq 2$ for $1\leq j\leq l$.

$(\romannumeral4)$   $|Y_{j}\cap C_{j-1}^{\ast}|\leq 1$ for $1\leq j\leq l$

$(\romannumeral5)$   $|Y_{j}|= 1$, for $j=l+1, \ldots, h$.
\end{lemma}
\begin{proof}
By Corollary \ref{cor14-5-9-2}, $G[C^{*}\cup C]$ is $3$-connected. It should be pointed out that all the following paths are taken in $G[C^{*}\cup C]$. The $3$-connectedness of $G[C^{*}\cup C]$ guarantees the existence of such paths.

Suppose $G[C]$ is not 3-connected. Let $P_{1}$ be a shortest brick-bridge of $G[C]$ in $G[C^*\cup C]$. By Observation \ref{obs1} and Lemma \ref{lemma5}, conditions $(\romannumeral1)$ to $(\romannumeral4)$ are satisfied for $j=1$.

Suppose that we have found subsets $Y_{1},\ldots,Y_{j}$ satisfying conditions $(\romannumeral1)$ to $(\romannumeral4)$ and $[C\cup C_j^*]$ is not 3-connected. Let $S$ be a 2-separator of $G[C\cup C_j^*]$ which is contained in $B^{(j)}$. As we have noticed by Corollary \ref{cor14-5-9-2}, any 2-separator of $G[C\cup C_j^*]$ is also a 2-separator of $G[C]$. Hence, if we denote by $\mathcal B_1$ the set of bricks of $G[C]$ contained in $B^{(j)}$ which also contain $S$, then $|\mathcal B_1|\geq 1$. Let $G_1$ be the union of those connected components of $G[C]-S$ containing $(\bigcup_{B\in\mathcal B_1}B)-S$, and let $G_2$ be the union of remaining connected components of $G[C]-S$. Similarly to the proof of Lemma \ref{lemma5}, a shortest path $P$ in $G[C^*\cup C]$ between $G_1$ and $G_2$ has at most two internal nodes. Notice that $B^{(j)}\cap C-S$ is contained in $G_1$. So $P$ contains a brick-bridge $P'$ of $G[C\cup C_j^*]$, and $B^{(j)}\in\mathcal Q_{int(P')}$. It follows that $B^{(j)}$ is contained in the new $T$-brick of $G[C\cup C_j^*\cup int(P)]$. Taking $P_{j+1}=P$, by Observation \ref{obs1}, conditions $(\romannumeral1)$ to $(\romannumeral4)$ are satisfied for $j+1$.

For $j\geq l$, it suffices to take $Y_{j+1}$ to be an arbitrary node in $C^{*}\backslash(C\cup C_{j}^{*})$.
\end{proof}

In the following proofs, condition $(\romannumeral1)$ of Lemma \ref{lemma6} is very important for a guaranteed performance ratio. The idea of condition $(\romannumeral1)$ is that when $Y_1,\ldots,Y_l$ are added sequentially, we are expanding ONE $T$-brick (instead of merging bricks here and there in a messy way), any brick of $G[C]$ which has empty intersection with this $T$-brick remains the same.

\begin{lemma}\label{lemma7}
Suppose $m\geq 3$, $C$ is a $(2,m)$-CDS of $G$, and $C^{\ast}$ is a minimum $(3,m)$-CDS of $G$. Let $C^{\ast}\backslash C=Y_{1}\cup Y_{2}\cup\ldots\cup Y_{h}$ be the decomposition as in Lemma \ref{lemma6}, and let $l$ be the first index such that $G[C\cup C_{l}^{\ast}]$ is 3-connected. Then for any $j=1,\ldots, l$,
\begin{equation}\label{eq1}
-\bigtriangleup_{Y_{j}}f(C\cup C_{j-1}^{\ast})\leq -\bigtriangleup_{Y_{j}}f(C)+6.
\end{equation}
Furthermore, if every $R$-brick of $G[C]$ has length three, then for any $j=1,\ldots, l$,
\begin{equation}\label{eq15-1-4-2}
-\bigtriangleup_{Y_{j}}f(C\cup C_{j-1}^{\ast})\leq -\bigtriangleup_{Y_{j}}f(C).
\end{equation}
\end{lemma}
\begin{proof}
The first part of the lemma is the result of the following two claims and the definition of $f$.

\vskip 0.2cm \noindent {\em Claim 1.} $\Delta_{Y_j}|\mathcal T(C)|-\Delta_{Y_j}|\mathcal T(C\cup C_{j-1}^*)|\leq 1$.
\vskip 0.2cm

In fact, by equation \eqref{eq14-6-21-1},
\begin{equation}\label{eq15-1-9-11}
\Delta_{Y_j}|\mathcal T(C)|-\Delta_{Y_j}|\mathcal T(C\cup C_{j-1}^*)|=|\mathcal Q_{Y_j}^T(C\cup C_{j-1}^*)|-|\mathcal Q_{Y_j}^T(C)|.
\end{equation}
Since $C_{j-1}^*$ is completely contained in one $T$-brick of $G[C\cup C_{j-1}^*]$ (see Lemma \ref{lemma6} $(\romannumeral1)$), we have $|\mathcal Q_{Y_j}^T(C\cup C_{j-1}^*)|-|\mathcal Q_{Y_j}^T(C)|\leq 1$. Claim 1 is proved.

\vskip 0.2cm\noindent {\em Claim 2.} $\Delta_{Y_j}q(C)-\Delta_{Y_j}q(C\cup C_{j-1}^*)\leq 5$.\vskip 0.2cm

The validity of Claim 2 is achieved by a series of sub-claims. The readers may refer to Fig.\ref{fig14-6-26-1} to help understanding the following proofs.

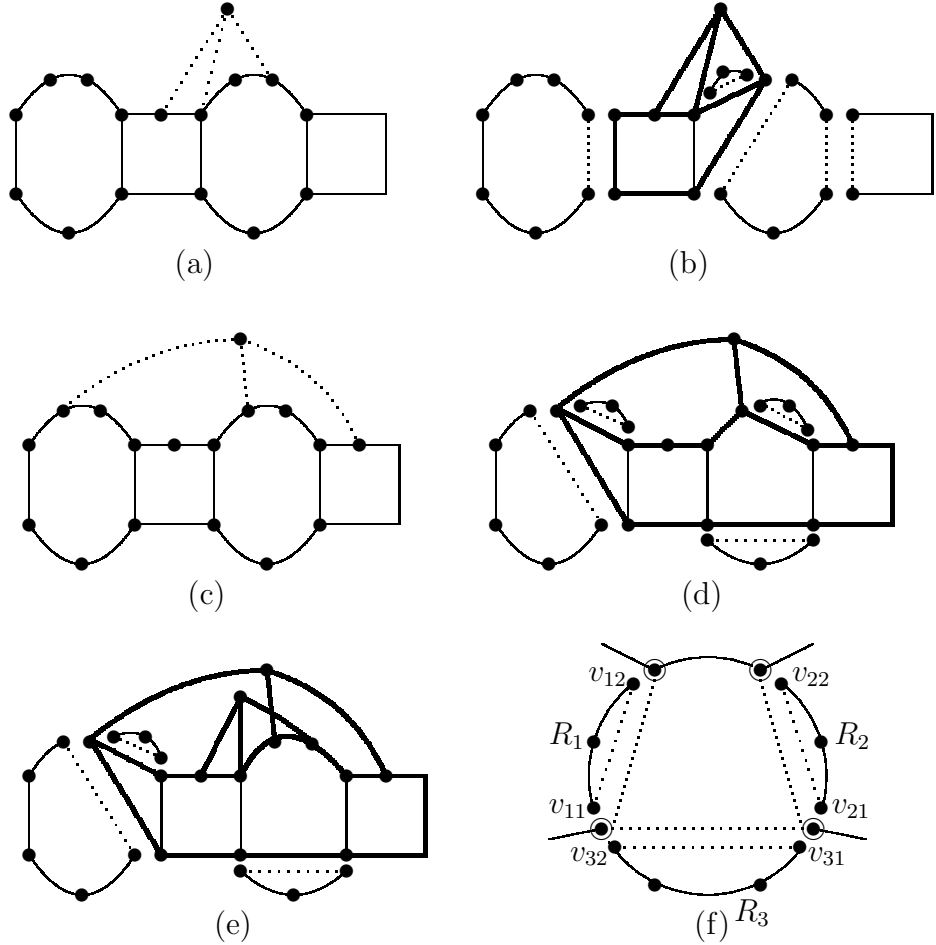
\begin{figure*}[!htbp]
\begin{center}
\begin{picture}(140,110)
\put(0,30){\circle*{5}}\put(40,30){\circle*{5}}\put(0,60){\circle*{5}}\put(40,60){\circle*{5}}
\put(70,30){\circle*{5}}\put(110,30){\circle*{5}}\put(70,60){\circle*{5}}\put(110,60){\circle*{5}}
\put(20,15){\circle*{5}}\put(13,73){\circle*{5}}\put(27,73){\circle*{5}}
\put(90,15){\circle*{5}}\put(83,73){\circle*{5}}\put(97,73){\circle*{5}}
\put(55,60){\circle*{5}}\put(80,100){\circle*{5}}
\qbezier(0,30)(0,45)(0,60)\qbezier(0,30)(20,0)(40,30)\qbezier(0,60)(20,90)(40,60)\qbezier(40,30)(40,45)(40,60)
\qbezier(40,30)(55,30)(70,30)\qbezier(40,60)(55,60)(70,60)\qbezier(70,30)(70,45)(70,60)
\qbezier(70,60)(90,90)(110,60)\qbezier(70,30)(90,0)(110,30)\qbezier(110,30)(110,45)(110,60)
\qbezier(110,30)(125,30)(140,30)\qbezier(110,60)(125,60)(140,60)\qbezier(140,30)(140,45)(140,60)
{\linethickness{0.25mm}\qbezier[13](80,100)(67.5,80)(55,60)\qbezier[10](80,100)(75,80)(70,60)\qbezier[9](80,100)(88.5,86.5)(97,73)}
\put(60,0){(a)}
\end{picture}
\hskip 1.5cm\begin{picture}(150,110)
\put(-10,30){\circle*{5}}\put(30,30){\circle*{5}}\put(40,30){\circle*{5}}
\put(-10,60){\circle*{5}}\put(30,60){\circle*{5}}\put(40,60){\circle*{5}}
\put(70,30){\circle*{5}}\put(80,30){\circle*{5}}\put(120,30){\circle*{5}}\put(70,60){\circle*{5}}\put(120,60){\circle*{5}}
\put(130,30){\circle*{5}}\put(130,60){\circle*{5}}
\put(10,15){\circle*{5}}\put(3,73){\circle*{5}}\put(17,73){\circle*{5}}
\put(76,68){\circle*{5}}\put(81,76){\circle*{5}}\put(90,75){\circle*{5}}
\put(100,15){\circle*{5}}\put(97,73){\circle*{5}}\put(107,73){\circle*{5}}
\put(55,60){\circle*{5}}\put(80,100){\circle*{5}}
\qbezier(-10,30)(-10,45)(-10,60)\qbezier(-10,30)(10,0)(30,30)\qbezier(-10,60)(10,90)(30,60)
\qbezier(40,30)(55,30)(70,30)\qbezier(40,60)(55,60)(70,60)\qbezier(70,30)(70,45)(70,60)
\qbezier(107,73)(114,70)(120,60)\qbezier(80,30)(100,0)(120,30)
\qbezier(130,30)(145,30)(160,30)\qbezier(130,60)(145,60)(160,60)\qbezier(160,30)(160,45)(160,60)
\qbezier(76,68)(81,80)(90,75)
{\linethickness{0.5mm}\qbezier(80,100)(67.5,80)(55,60)
\qbezier(80,100)(75,80)(70,60)\qbezier(80,100)(88.5,86.5)(97,73)
\qbezier(97,73)(83.5,66.5)(70,60)\qbezier(97,73)(83.5,51.5)(70,30)
\qbezier(40,30)(55,30)(70,30)\qbezier(40,60)(55,60)(70,60)
\qbezier(40,30)(40,45)(40,60) }
{\linethickness{0.25mm}\qbezier[17](107,73)(93.5,51.5)(80,30)\qbezier[6](76,68)(82,72)(90,75)
\qbezier[11](30,30)(30,45)(30,60)\qbezier[11](120,30)(120,45)(120,60)\qbezier[11](130,30)(130,45)(130,60)}
\put(60,0){(b)}
\end{picture}
\vskip 0.5cm
\begin{picture}(140,110)
\put(0,30){\circle*{5}}\put(40,30){\circle*{5}}\put(0,60){\circle*{5}}\put(40,60){\circle*{5}}
\put(70,30){\circle*{5}}\put(110,30){\circle*{5}}\put(70,60){\circle*{5}}\put(110,60){\circle*{5}}
\put(20,15){\circle*{5}}\put(13,73){\circle*{5}}\put(27,73){\circle*{5}}
\put(90,15){\circle*{5}}\put(83,73){\circle*{5}}\put(97,73){\circle*{5}}
\put(55,60){\circle*{5}}\put(125,60){\circle*{5}}\put(80,100){\circle*{5}}
\qbezier(0,30)(0,45)(0,60)\qbezier(0,30)(20,0)(40,30)\qbezier(0,60)(20,90)(40,60)\qbezier(40,30)(40,45)(40,60)
\qbezier(40,30)(55,30)(70,30)\qbezier(40,60)(55,60)(70,60)\qbezier(70,30)(70,45)(70,60)
\qbezier(70,60)(90,90)(110,60)\qbezier(70,30)(90,0)(110,30)\qbezier(110,30)(110,45)(110,60)
\qbezier(110,30)(125,30)(140,30)\qbezier(110,60)(125,60)(140,60)\qbezier(140,30)(140,45)(140,60)
{\linethickness{0.25mm}\qbezier[23](80,100)(45,100)(13,74)\qbezier[9](80,100)(81.5,86.5)(83,73)\qbezier[19](80,100)(112,90)(125,60)}
\put(60,0){(c)}
\end{picture}
\hskip 1.5cm
\begin{picture}(140,110)
\put(-10,30){\circle*{5}}\put(30,30){\circle*{5}}\put(40,30){\circle*{5}}\put(-10,60){\circle*{5}}\put(40,60){\circle*{5}}
\put(70,30){\circle*{5}}\put(110,30){\circle*{5}}\put(70,24){\circle*{5}}\put(110,24){\circle*{5}}\put(70,60){\circle*{5}}\put(110,60){\circle*{5}}
\put(22,75){\circle*{5}}\put(34,75){\circle*{5}}\put(40,67){\circle*{5}}
\put(10,15){\circle*{5}}\put(3,73){\circle*{5}}\put(13,73){\circle*{5}}
\put(90,15){\circle*{5}}\put(83,73){\circle*{5}}
\put(90,75){\circle*{5}}\put(101,75){\circle*{5}}\put(108,66){\circle*{5}}
\put(55,60){\circle*{5}}\put(125,60){\circle*{5}}\put(80,100){\circle*{5}}
\qbezier(-10,30)(-10,45)(-10,60)\qbezier(-10,30)(10,0)(30,30)\qbezier(-10,60)(-4,70)(3,73)\qbezier(40,30)(40,45)(40,60)
\qbezier(70,30)(70,45)(70,60)\qbezier(110,30)(110,45)(110,60)
\qbezier(70,24)(90,6)(110,24)\qbezier(22,75)(35,80)(40,67)\qbezier(90,76)(101,81)(108,66)
{\linethickness{0.5mm}\qbezier(80,100)(45,100)(13,74)\qbezier(80,100)(81.5,86.5)(83,73)\qbezier(80,100)(112,90)(125,60)
\qbezier(40,60)(26.5,67)(13,74)\qbezier(40,30)(26.5,52)(13,74)\qbezier(40,60)(55,60)(70,60)\qbezier(110,60)(125,60)(140,60)
\qbezier(40,30)(90,30)(140,30)\qbezier(110,60)(125,60)(140,60)\qbezier(140,30)(140,45)(140,60)
\qbezier(70,60)(76.5,66.5)(83,73)\qbezier(110,60)(96.5,66.5)(83,73) }
{\linethickness{0.25mm}\qbezier[15](30,30)(16.5,51.5)(3,73)\qbezier[11](70,24)(90,24)(110,24)
\qbezier[7](22,75)(31,71)(40,67)\qbezier[8](90,76)(99,71)(108,66)}
\put(60,0){(d)}
\end{picture}
\vskip 0.5cm
\begin{picture}(140,110)
\put(-10,30){\circle*{5}}\put(30,30){\circle*{5}}\put(40,30){\circle*{5}}\put(-10,60){\circle*{5}}\put(40,60){\circle*{5}}
\put(70,30){\circle*{5}}\put(110,30){\circle*{5}}\put(70,24){\circle*{5}}\put(110,24){\circle*{5}}\put(70,60){\circle*{5}}\put(110,60){\circle*{5}}
\put(22,75){\circle*{5}}\put(34,75){\circle*{5}}\put(40,67){\circle*{5}}
\put(10,15){\circle*{5}}\put(3,73){\circle*{5}}\put(13,73){\circle*{5}}
\put(90,15){\circle*{5}}\put(83,73){\circle*{5}}\put(97,72){\circle*{5}}
\put(55,60){\circle*{5}}\put(125,60){\circle*{5}}\put(80,100){\circle*{5}}\put(70,90){\circle*{5}}
\qbezier(-10,30)(-10,45)(-10,60)\qbezier(-10,30)(10,0)(30,30)\qbezier(-10,60)(-4,70)(3,73)\qbezier(40,30)(40,45)(40,60)
\qbezier(70,30)(70,45)(70,60)\qbezier(110,30)(110,45)(110,60)\qbezier(70,60)(85,90)(110,60)
\qbezier(70,24)(90,6)(110,24)\qbezier(22,75)(35,80)(40,67)
{\linethickness{0.5mm}\qbezier(80,100)(45,100)(13,74)\qbezier(80,100)(81.5,86.5)(83,73)\qbezier(80,100)(112,90)(125,60)
\qbezier(40,60)(26.5,67)(13,74)\qbezier(40,30)(26.5,52)(13,74)\qbezier(40,60)(55,60)(70,60)\qbezier(110,60)(125,60)(140,60)
\qbezier(40,30)(90,30)(140,30)\qbezier(110,60)(125,60)(140,60)\qbezier(140,30)(140,45)(140,60)\qbezier(70,60)(85,90)(110,60)
\qbezier(70,90)(62.5,75)(55,60)\qbezier(70,90)(70,75)(70,60)\qbezier(70,90)(83.5,84)(97,73) }
{\linethickness{0.25mm}\qbezier[15](30,30)(16.5,51.5)(3,73)\qbezier[11](70,24)(90,24)(110,24)\qbezier[7](22,75)(31,71)(40,67)}
\put(60,0){(e)}
\end{picture}
\hskip 1.5cm
\begin{picture}(120,110)
\put(40,100){\circle*{5}}\put(80,100){\circle*{5}}\put(20,40){\circle*{5}}\put(100,40){\circle*{5}}
\put(40,100){\circle{8}}\put(80,100){\circle{8}}\put(20,40){\circle{8}}\put(100,40){\circle{8}}
\put(17,48){\circle*{5}}\put(17,73){\circle*{5}}\put(32,95){\circle*{5}}
\put(103,48){\circle*{5}}\put(103,73){\circle*{5}}\put(88,95){\circle*{5}}
\put(25,33){\circle*{5}}\put(40,19){\circle*{5}}\put(95,33){\circle*{5}}\put(80,19){\circle*{5}}
\qbezier(40,100)(30,105)(20,110)\qbezier(80,100)(90,105)(100,110)\qbezier(20,40)(10,38)(0,36)\qbezier(100,40)(110,38)(120,36)

{\linethickness{0.25mm}\qbezier[19](25,33)(60,33)(95,33)\qbezier[13](17,48)(24.5,71.5)(32,95)\qbezier[13](103,48)(95.5,71.5)(88,95)
\qbezier[21](40,96)(32,68)(24,40)\qbezier[21](80,96)(88,68)(96,40)\qbezier[21](24,40)(60,40)(96,40)}

\put(0,72){$R_1$}\put(108,72){$R_2$}\put(70,5){$R_3$}
\put(0,46){$v_{11}$}\put(15,96){$v_{12}$}\put(107,46){$v_{21}$}\put(92,96){$v_{22}$}\put(98,28){$v_{31}$}\put(8,28){$v_{32}$}

\qbezier(105.0,60.0)(104.8,63.8)(104.4,67.6)
\qbezier(104.4,67.6)(103.5,71.3)(102.4,75.0)
\qbezier(102.4,75.0)(101.0,78.5)(99.3,81.9)
\qbezier(99.3,81.9)(97.3,85.2)(95.0,88.3)
\qbezier(95.0,88.3)(92.5,91.1)(89.7,93.8)
%\qbezier(89.7,93.8)(86.8,96.2)(83.6,98.3)
%\qbezier(83.6,98.3)(80.3,100.2)(76.8,101.8)
\qbezier(76.8,101.8)(73.2,103.0)(69.5,104.0)
\qbezier(69.5,104.0)(65.7,104.6)(61.9,105.0)
\qbezier(61.9,105.0)(58.1,105.0)(54.3,104.6)
\qbezier(54.3,104.6)(50.5,104.0)(46.8,103.0)
\qbezier(46.8,103.0)(43.2,101.8)(39.7,100.2)
%\qbezier(39.7,100.2)(36.4,98.3)(33.2,96.2)
\qbezier(33.2,96.2)(30.3,93.8)(27.5,91.1)
\qbezier(27.5,91.1)(25.0,88.3)(22.7,85.2)
\qbezier(22.7,85.2)(20.7,81.9)(19.0,78.5)
\qbezier(19.0,78.5)(17.6,75.0)(16.5,71.3)
\qbezier(16.5,71.3)(15.6,67.6)(15.2,63.8)
\qbezier(15.2,63.8)(15.0,60.0)(15.2,56.2)
\qbezier(15.2,56.2)(15.6,52.4)(16.5,48.7)
%\qbezier(16.5,48.7)(17.6,45.0)(19.0,41.5)
%\qbezier(19.0,41.5)(20.7,38.1)(22.7,34.8)
\qbezier(22.7,34.8)(25.0,31.7)(27.5,28.9)
\qbezier(27.5,28.9)(30.3,26.2)(33.2,23.8)
\qbezier(33.2,23.8)(36.4,21.7)(39.7,19.8)
\qbezier(39.7,19.8)(43.2,18.2)(46.8,17.0)
\qbezier(46.8,17.0)(50.5,16.0)(54.3,15.4)
\qbezier(54.3,15.4)(58.1,15.0)(61.9,15.0)
\qbezier(61.9,15.0)(65.7,15.4)(69.5,16.0)
\qbezier(69.5,16.0)(73.2,17.0)(76.8,18.2)
\qbezier(76.8,18.2)(80.3,19.8)(83.6,21.7)
\qbezier(83.6,21.7)(86.8,23.8)(89.7,26.2)
\qbezier(89.7,26.2)(92.5,28.9)(95.0,31.7)
%\qbezier(95.0,31.7)(97.3,34.8)(99.3,38.1)
%\qbezier(99.3,38.1)(101.0,41.5)(102.4,45.0)
%\qbezier(102.4,45.0)(103.5,48.7)(104.4,52.4)
\qbezier(103.5,48.7)(103.5,48.7)(104.4,52.4)
\qbezier(104.4,52.4)(104.8,56.2)(105.0,60.0)
\put(55,0){(f)}
\end{picture}
\caption{(a) The solid lines indicate $G[C]$. Each rectangle represents a $T$-brick. Each rounded rectangle represents an $R$-brick. Together with the dashed lines, we have $G[C\cup C_{j-1}^*]$. (b) is the brick decomposition of $G[C\cup C_{j-1}^*]$. The blackened lines indicate $B^{(j-1)}$. (c) depicts $G[C\cup Y_j]$, the dashed lines are the edges added together with the addition of $Y_j$. (d) is the brick decomposition of $G[C\cup Y_j]$. The blackened lines indicate $B$. (e) is the brick decomposition of $G[C\cup C_{j-1}^*\cup Y_j]=G[C\cup C_j^*]$. The blackened lines indicate $B^{(j)}$. In (f), an $R$-brick is divided into smaller $R$-bricks of $G[C\cup Y_j]$ by the new $T$-brick $B$. The double circled nodes are in $V(B)\cap V(R)$. The center part belongs to the new $T$-brick $B$. The top arc belongs to $E(B)\cap E(R)$. The dashed lines are virtual edges.}\label{fig14-6-26-1}
\end{center}
\end{figure*}

By the definition of $q$,
\begin{align}\label{eq14-6-23-12}
\Delta_{Y_j}q(C)-\Delta_{Y_j}q(C\cup C_{j-1}^*)=\Delta_{Y_j}\sum_{R\in \mathcal R(C)}(2|R|-5)-\Delta_{Y_j}\!\!\!\!\sum_{R\in\mathcal R(C\cup C_{j-1}^*)}\!\!\!\!(2|R|-5).
\end{align}

Let $R$ be an $R$-brick of $G[C]$. If $R$ contributes to the first term of \eqref{eq14-6-23-12}, then by Observation \ref{obs1}, $R$ is divided by the new $T$-brick $B$ of $G[C\cup Y_j]$ containing $Y_j$. By Lemma \ref{lemma6} $(\romannumeral1)$, it can be seen that

\vskip 0.2cm\noindent {\em SubClaim 2.1.} $V(B^{(j)})\cap V(R)=\big(V(B^{(j-1)})\cup V(B)\big)\cap V(R).$\vskip 0.2cm

%\begin{equation}\label{eq14-9-3-1}
%V(B^{(j)})\cap V(R)=\big(V(B^{(j-1)})\cup V(B)\big)\cap V(R).
%\end{equation}

%As a consequence,
%\begin{equation}\label{eq14-9-1-1}
%V(B)\cap V(R)\subseteq V(B^{(j)})\cap V(R),
%\end{equation}
%and
%\begin{equation}\label{eq14-9-1-2}
%(V(B^{(j)})\setminus V(B^{(j-1)}))\cap V(R)\subseteq V(B)\cap V(R).
%\end{equation}

As in the proof of Lemma \ref{lemma11}, denote by $\mathcal R_{C,X}^{div}(R)$ the set of smaller $R$-bricks of $G[C\cup X]$ arising from the division of $R$ after $X$ is added into $C$.

\vskip 0.2cm\noindent {\em SubClaim 2.2.} $\sum\limits_{R'\in\mathcal R_{C,Y_j}^{div}(R)}(2|R'|-5)-(2|R|-5)=
|E(B)\cap E(R)|-3|V(B)\cap V(R)|+5$.\vskip 0.2cm

For simplicity of statement, suppose $s=|\mathcal R_{C,Y_j}^{div}(R)|$ and $\mathcal R_{C,Y_j}^{div}(R)=\{R_1,\ldots,R_s\}$ (see Fig.\ref{fig14-6-26-1}(f) for an illustration). For $i=1,\ldots,s$, denote by $S_i=V(R_i)\cap V(B)=\{v_{i1},v_{i2}\}$. By Observation \ref{obs1}, the subgraph of $B$ induced by $(E(B)\cap E(R))\cup\{v_{11}v_{12},v_{21}v_{22},
\ldots,v_{s1}v_{s2}\}$ is a cycle. So,
\begin{equation}\label{eq14-8-31-1}
|V(R)|=\sum_{i=1}^{s}
|R_i|+|V(B)\cap V(R)|-2s
\end{equation}
and
\begin{equation}\label{eq14-8-31-2}
s+|E(B)\cap E(R)|=|V(B)\cap V(R)|.
\end{equation}
It follows that
\begin{align}\label{equ5}
&\sum_{i=1}^{s}(2|R_{i}|-5)\nonumber\\
=\ &2(|V(R)|-|V(B)\cap V(R)|+2s)-5s\nonumber\\
=\ &2(|V(R)|-|V(B)\cap V(R)|) \nonumber\\
&-(|V(B)\cap V(R)|-|E(B)\cap E(R)|)\nonumber\\
=\ &2|R|-3|V(B)\cap V(R)|+|E(B)\cap E(R)|.
\end{align}
Then, SubClaim 2.2 follows.

Notice that SubClaim 2.2 provides an expression for each $R\in\mathcal R(C)$ in the first term of the righthand side of \eqref{eq14-6-23-12}. Estimation on the second term of the righthand side of \eqref{eq14-6-23-12} can make use of SubClaim 2.2. In fact, consider those $R$-bricks in $\mathcal R^{div}_{C,C_{j-1}^*}(R)$ (where $R$ is the $R$-brick in SubClaim 2.2), they are further divided into smaller $R$-bricks when $Y_j$ is added into $C\cup C_{j-1}^*$ (see Fig.\ref{fig14-6-26-1}(b) and (e)).
% noticing that $\mathcal R^{div}_{C,C_j^*}(R)=\mathcal R^{div}_{C\cup C_{j-1}^*,Y_j}(R)$, and
Making use of SubClaim 2.2 (replacing $Y_j$ by $C_j^*$ and $C_{j-1}^*$, and replacing $B$ by $B^{(j)}$ and $B^{(j-1)}$, correspondingly), it can be estimated that
\begin{align}\label{equ7}
&\Delta_{Y_j}\sum_{R'\in\mathcal R^{div}_{C,C_{j-1}^*}(R)}(2|R'|-5)\nonumber\\
=\ &\left(\sum_{R'\in \mathcal R^{div}_{C,C_{j}^{*}}(R)}(2|R'|-5)-(2|R|-5)\right)-\left(\sum_{R'\in\mathcal R^{div}_{C,C_{j-1}^{*}}(R)}(2|R'|-5)-(2|R|-5)\right)\nonumber\\
=\ &|(E(B^{(j)})\backslash E(B^{(j-1)}))\cap E(R)|-3|(V(B^{(j)})\backslash V(B^{(j-1)}))\cap V(R)|.
\end{align}
So, for each $R\in\mathcal R(C)$, if we denote by $g(R)$ the total value of those terms in the righthand side of \eqref{eq14-6-23-12} which are related with $R$, then by SubClaim 2.2 and \eqref{equ7}, it can be seen that $g(R)$ has the following expression:
\begin{eqnarray}\label{equ8}
g(R)&=&5+3\big(|(V(B^{(j)})\backslash V(B^{(j-1)}))\cap V(R)|-|V(B)\cap V(R)|\big)\nonumber\\
&&+\big(|E(B)\cap E(R)|-|(E(B^{(j)})\backslash E(B^{(j-1)}))\cap E(R)|\big).
\end{eqnarray}
Notice that \eqref{eq14-6-23-12} can be rewritten as the following:

\vskip 0.2cm\noindent {\em SubClaim 2.3.} $\Delta_{Y_j}q(C)-\Delta_{Y_j}q(C\cup C_{j-1}^*)=\sum_{R\in \mathcal A}g(R)$, where $\mathcal A=\{R\in\mathcal R(C)\setminus \mathcal R(C\cup C_{j-1}^*)\colon R\ \mbox{is divided by}\ B\}$.\vskip 0.2cm

%\begin{equation}\label{eq14-9-2-1}
%\Delta_{Y_j}q(C)-\Delta_{Y_j}q(C\cup C_{j-1}^*)=\sum_{R\in \mathcal A}g(R),
%\end{equation}
The reason why only those $R$-bricks in $\mathcal R(C)\setminus\mathcal R(C\cup C_{j-1}^*)$ are considered is as follows: If $R\in \mathcal R(C)\cap \mathcal R(C\cup C_{j-1}^{*})$, then the changes on $R$ are the same in the two terms of \eqref{eq14-6-23-12}, which will cancel. The reason why only those $R$-bricks divided by $B$ are considered is the following: for any $R$-brick $R$ which is not divided by $B$, adding $Y_j$ does not change $R$, neither does it change any smaller $R$-bricks in $\mathcal R^{div}_{C,C_{j-1}^*}(R)$.

The next subclaim estimates the upper bound for $g(R)$. Suppose $|V(B)\cap V(R)|-|(V(B^{(j)})\setminus V(B^{(j-1)}))\cap V(R)|= t(R)$.

\vskip 0.2cm\noindent {\em SubClaim 2.4.}\\
 $g(R)\leq \left\{\begin{array}{ll} 5, &  t(R)=0,\\ -1, & t(R)\geq 1\ \mbox{and \eqref{eq14-9-3-4} occurs},\\ 5-2t(R)-1, & t(R)\geq 1\ \mbox{and \eqref{eq14-9-3-4} does not occur.}\end{array}\right.$\vskip 0.2cm

By SubClaim 2.1, it can be seen that $t(R)$ can be rewritten as $t(R)=|V(B)\cap V(B^{(j-1)})\cap V(R)|$. So, $t(R)\geq 0$. If $t(R)=0$, then $V(B)\cap V(B^{(j-1)})\cap V(R)=\emptyset$, and thus $\big(E(B)\cap E(R)\big)\cap\big(E(B^{(j-1)})\cap E(R)\big)=\emptyset$. By noticing that $V(B)\cap V(R)\subseteq V(B^{(j)})\cap V(R)$ by SubClaim 2.1, and thus $E(B)\cap E(R)\subseteq E(B^{(j)})\cap E(R)$, we have $E(B)\cap E(R)\subseteq (E(B^{(j)})\setminus E(B^{(j-1)}))\cap E(R)$, and thus $g(R)\leq 5$ by \eqref{equ8}. When $t(R)\geq 1$, by recalling that $R$ is a cycle, we see that  $|E(B)\cap E(B^{(j-1)})\cap E(R)|\leq t(R)-1$ unless
\begin{align}\label{eq14-9-3-4}
V(B)\cap V(R)&=V(B^{(j-1)})\cap V(R)\nonumber\\
&=V(B^{(j)})\cap V(R)=V(R).
\end{align}
If \eqref{eq14-9-3-4} occurs, then we see from \eqref{equ8} that $g(R)=5-2|R|\leq -1$. Otherwise, $|E(B)\cap E(R)|- |E(B^{(j)})\setminus E(B^{(j-1)}))\cap E(R)|=|E(B)\cap E(B^{(j-1)})\cap E(R)|\leq t(R)-1$ and thus $g(R)\leq 5-2t(R)-1$ by \eqref{equ8}. SubClaim 2.4 is proved.

\vskip 0.2cm\noindent {\em SubClaim 2.5.} If $|\mathcal A|\geq 2$, then for every $R\in \mathcal A$, $t(R)\geq 2$. \vskip 0.2cm

In fact, since such an $R$-brick does not belong to $\mathcal R(C\cup C_{j-1}^*)$, it is divided by $B^{(j-1)}$. For $R,R'\in\mathcal A$, consider the unique path $Q_{RR'}$ in the brick tree of $G[C]$ connecting $R$ and $R'$, the first 2-separator incident with $R$, say $S$, must belong to $V(B^{(j-1)})$ (by Observation \ref{obs1} ($\romannumeral4$)). Since both $R$ and $R'$ are divided by $B$, for the same reason, $S\subseteq V(B)$. So, $t(R)=|V(B)\cap V(B^{(j-1)})\cap V(R)|\geq |S|=2$. SubClaim 2.5 is proved.

\vskip 0.2cm Combing SubClaim 2.4 and SubClaim 2.5, if $|\mathcal A|\geq 2$, then $g(R)\leq 0$ for any $R\in\mathcal A$. Otherwise, $g(R)=0$ if $\mathcal A=\emptyset$ and $g(R)\leq 5$ if $|\mathcal A|=1$. Them Claim 2 follows from SubClaim 2.3.

Combining Claim 1 and Claim 2, the first part of this lemma is proved.

In the case that every $R$-brick of $G[C]$ has length three, we see from Observation \ref{obs1} $(\romannumeral2)$ that after adding a node set, any $R$-brick either diminishes or remains the same. Denote by $\mathcal R^{dim}_{C,X}$ the set of $R$-bricks diminished after adding $X$ into $C$. By Lemma \ref{lemma6} $(\romannumeral1)$ and $(\romannumeral2)$, we see that $\mathcal R^{dim}_{C\cup C_{j-1}^*,Y_j}\subseteq \mathcal R^{dim}_{C,Y_j}$. Hence
\begin{equation}\label{eq15-1-9-3}
\Delta_{Y_j}q(C)-\Delta_{Y_j}q(C\cup C_{j-1}^*)=-\sum_{R\in\mathcal R^{dim}_{C,Y_j}\setminus\mathcal R^{dim}_{C\cup C_{j-1}^*,Y_j}}(2|R|-5)\leq 0.
\end{equation}
Furthermore, if $\Delta_{Y_j}|\mathcal T(C)|-\Delta_{Y_j}|\mathcal T(C\cup C_{j-1}^*)|=1$, then $|\mathcal Q_{Y_j}^T(C\cup C_{j-1}^*)|-|\mathcal Q_{Y_j}^T(C)|=1$ by \eqref{eq15-1-9-11}, which is possible only when adding $C_{j-1}^*$ into $C$ creates a new $T$-brick, which occurs only when every brick in $\mathcal Q_{C_{j-1}^*}(C)$ is an $R$-brick. Combining this with Lemma \ref{lemma6} $(\romannumeral2)$, we see that $\mathcal R^{dim}_{C,Y_j}\setminus\mathcal R^{dim}_{C\cup C_{j-1}^*,Y_j}$ contains at least one $R$-brick, and thus inequality \eqref{eq15-1-9-3} becomes
$$
\Delta_{Y_j}q(C)-\Delta_{Y_j}q(C\cup C_{j-1}^*)=-\sum_{R\in\mathcal R^{dim}_{C,Y_j}\setminus\mathcal R^{dim}_{C\cup C_{j-1}^*,Y_j}}(2|R|-5)\leq -1.
$$
Then the second part of this lemma follows from the definition of $f$.
\end{proof}

In the following, we use $X_1,X_2,\ldots,X_g$ to denote the sets chosen by Algorithm \ref{algo1}, in the order of their selection into set $C$. For $1\leq i \leq g$, denote $C_i =C_0\cup X_1\cup X_2, \cdots,\cup X_i$.

\begin{lemma}\label{le15-1-7-1}
For $1\leq i \leq g$, we have $\frac{-\bigtriangleup_{X_{i}}f(C_{i-1})}{|X_{i}|}
\geq\frac{1}{2}$. Furthermore, if $\mathcal R(C_{i-1})$ contains at least one $R$-brick of length at least $4$, then $\frac{-\bigtriangleup_{X_{i}}f(C_{i-1})}{|X_{i}|}
\geq1$.
\end{lemma}
\begin{proof}
For $1\leq i\leq g$, we have $-\bigtriangleup_{X_{i}}f(C_{i-1})\geq1$ by Lemma \ref{lemma11}. Then $\frac{-\bigtriangleup_{X_{i}}f(C_{i-1})}{|X_{i}|}\geq\frac{1}{2}$ follows from $|X_{i}|\leq2$.

If $G[C_{i-1}]$ is not a cycle and $\mathcal R(C_{i-1})$ contains at least one $R$-brick of length at least $4$, then by Lemma \ref{lemma5}, there exists a brick-bridge $P$ with $X=int(P)$ such that $|\mathcal Q_{X}|\geq2$, $|X|\leq2$ and $R\in \mathcal Q_{X}$. By Lemma \ref{lemma11} and the greedy rule of Algorithm \ref{algo1}, we have $-\bigtriangleup_{X}f(C_{i-1})\geq2$ and $\frac{-\bigtriangleup_{X_{i}}f(C_{i-1})}{|X_{i}|}
\geq\frac{-\bigtriangleup_{X}f(C_{i-1})}{|X|}\geq1$.

Notice that $G[C_{i}]$ cannot be a cycle for $i>0$. Recall that the case that $G[C_{0}]$ is a triangle is dealt with separately in Algorithm \ref{algo1}. In the case that $G[C_{0}]$ is a cycle of length at least $4$, consider an arbitrary node $v\in V(G)\backslash C_{0}$. Since $C_{0}$ is a $(2,m)$-$CDS$ and $m\geq3$, node $v$ must have two neighbors $u_{1},u_{2}$ in $C_{0}$ which are not consecutive on cycle $G[C_{0}]$. Let $P=u_{1}vu_{2}$. Then $P$ is a brick-bridge of $C_{0}$ and $\frac{-\bigtriangleup_{X_{1}}f(C_{0})}{|X_{1}|}
\geq\frac{-\bigtriangleup_{\{v\}}f(C_{0})}{|\{v\}|}\geq1$. The lemma is proved.
\end{proof}

Now, we are ready to prove the performance ratio.

\begin{theorem}\label{th-0}
Algorithm \ref{algo1} is a polynomial-time $\gamma$-approximation for the minimum $(3,m)$-CDS problem, where $\gamma=(3\alpha+2\ln2)$ for $\alpha<4$ and $\gamma=(\alpha+8+2\ln(2\alpha-6))$ for $\alpha\geq4$, $\alpha$ is the performance ratio for the minimum $(2,m)$-CDS problem.
\end{theorem}

\begin{proof}
By Corollary \ref{cor14-5-9-2}, every $C_{i}$ is a $(2,m)$-CDS for $0\leq i \leq g$. Suppose $q$ is the first index such that $\mathcal R(C_{q})$ contains no $R$-brick of length at least four. Let $C^{*}$ be a minimum $(3,m)$-CDS of $G$. Denote $|C^{*}|=t$.

\vskip 0.2cm \noindent {\em Claim 1.} $|C_{0}|\leq\alpha t$.
\vskip 0.2cm

Since $C_{0}$ is an $\alpha$-approximation for the minimum $(2,m)$-CDS problem, and because the size of a minimum $(2,m)$-CDS is no greater than the size of a minimum $(3,m)$-CDS, the claim follows.

For $0\leq i\leq g$, denote $a_{i}=f(C_{i})-6t-1$ and $b_{i}=f(C_{i})-1$.

\vskip 0.2cm\noindent {\em Claim 2.}\\
$|X_{i+1}|\leq\left\{\begin{array}{ll} \min\{a_{i}-a_{i+1},2t\frac{a_{i}-a_{i+1}}{a_{i}}\}, & \mbox{for}\ 0\leq i\leq q\!-\!1,\\  \min\{\!2(b_{i}\!-\!\!b_{i+1}\!),\!2t\!\frac{b_{i}-b_{i+1}}{b_{i}}\} , & \mbox{for}\ q\leq i\leq g\!-\!1.\end{array}\right.$

For any fixed $i$ with $0\leq i\leq g-1$, decompose $C^{\ast}\backslash C_i$ into $Y_{1}^{(i)},Y_{2}^{(i)},\ldots,Y_{h_{i}}^{(i)}$ satisfying those conditions of Lemma \ref{lemma6}. For $1\leq j\leq h_{i}$, denote $C_{j}^{\ast}=Y_{1}^{(i)}\cup Y_{2}^{(i)},\ldots,\cup Y_{j}^{(i)}$.
Set $C_{0}^{\ast}=\emptyset$. Suppose $l_{i}$ is the first index such that $G[C_{i}\cup C_{l_{i}}^{\ast}]$ is $3$-connected.

First, consider $C_{i}$ with $i=0,1,\ldots,q-1$. By Lemma \ref{lemma7}, for $1\leq j\leq l_{i},$
\begin{align}\label{eq13}
-\bigtriangleup_{Y_{j}}f(C\cup C_{j-1}^{\ast})\leq -\bigtriangleup_{Y_{j}}f(C)+6.
\end{align}
By the greedy rule of Algorithm \ref{algo1}, we have
\begin{equation}\label{eq14}
\frac{-\bigtriangleup_{X_{i+1}}f(C_{i})}{|X_{i+1}|}\geq \frac{-\bigtriangleup_{Y_{j}^{(i)}}
f(C_{i})}{|Y_{j}^{(i)}|},\ \mbox{for}\ j=1,2,\ldots,l_{i}.
\end{equation}
By Lemma \ref{lemma6},
\begin{equation}\label{eq15}
\sum_{j=1}^{l_{i}}|Y_{j}^{(i)}|
\leq |C^{\ast}\backslash C_{i}|+l_{i}
\leq 2t.
\end{equation}
Combing inequalities (\ref{eq13}),(\ref{eq14}),(\ref{eq15}) with the assumption that $G[C_{i}\cup C_{l_{i}}^{\ast}]$ is $3$-connected (and thus $f(C_{i}\cup C_{l_{i}}^{\ast})=1$ by Lemma \ref{lem14-6-19-6}), we have \begin{align}
\frac{-\bigtriangleup_{X_{i+1}}f(C_{i})}
{|X_{i+1}|}
&\geq\frac{-\sum_{j=1}^{l_{i}}
\bigtriangleup_{Y_{j}^{(i)}}
f(C_{i})}{\sum_{j=1}^{l_{i}}
|Y_{j}^{(i)}|}\notag\\
&\geq \frac{\sum_{j=1}^{l_{i}}
(-\bigtriangleup_{Y_{j}^{(i)}}f(C_{i}\cup C_{j-1}^{\ast})-6)}{2t}\notag\\
&=\frac{-(f(C_{i}\cup C_{l_{i}}^{\ast})-f(C_{i}))-6l_{i}}{2t}\notag\\
&\geq\frac{-(f(C_{i}\cup C_{l_{i}}^{\ast})-f(C_{i}))-6t}{2t}\notag\\
&=\frac{f(C_{i})-6t-1}{2t}.
\end{align}
The above inequality can be rewritten as

\begin{equation}\label{eq16}
\frac{a_{i}-a_{i+1}}{|X_{i+1}|}\geq \frac{a_{i}}{2t},\ \mbox{for}\ 0\leq i \leq q-1.
 \end{equation}
 and thus
\begin{equation}\label{eq17}
|X_{i+1}|\leq 2t\frac{a_{i}-a_{i+1}}{a_{i}},\ \mbox{for}\ 0\leq i \leq q-1.
\end{equation}

Next, consider $C_{i}$ with $q\leq i\leq g-1$. By the second part of Lemma \ref{lemma7}, we have
\begin{align}\label{eq15-1-4-3}
-\bigtriangleup_{Y_{j}}f(C\cup C_{j-1}^{\ast})\leq -\bigtriangleup_{Y_{j}}f(C).
\end{align}
Similar to the derivation of inequalities (\ref{eq16}) and (\ref{eq17}), we have
\begin{equation}\label{eq15-1-4-4}
\frac{b_{i}-b_{i+1}}{|X_{i+1}|}\geq \frac{b_{i}}{2t}\ \mbox{for}\ q\leq i \leq g-1
 \end{equation}
 and
\begin{equation}\label{eq15-1-4-5}
|X_{i+1}|\leq 2t\frac{b_{i}-b_{i+1}}{b_{i}}\ \mbox{for}\ q\leq i \leq g-1.
\end{equation}

By Lemma \ref{lemma11},
\begin{equation}\label{eq15-1-10-1}
\frac{f(C_{i})-f(C_{i+1})}{|X_{i+1}|}\geq \left\{\begin{array}{ll} 1, & \mbox{for}\ 0\leq i\leq q-1,\\ 1/2, & \mbox{for}\ q\leq i\leq g-1.\end{array}\right.
\end{equation}
By the definition of $a_{i}$ and $b_{i}$, $f(C_{i})-f(C_{i+1})=a_{i}-a_{i+1}=b_{i}-b_{i+1}$. So
$|X_{i+1}|\leq a_{i}-a_{i+1}$ for $0\leq i\leq q-1$ and $|X_{i+1}|\leq 2(b_{i}-b_{i+1})$ for $q\leq i\leq g-1$.
Claim 2 is proved.

\vskip 0.2cm \noindent {\em Claim 3.} If $a_{0}\geq2t$, then $\sum_{i=0}^{g-1}|X_{i+1}|\leq8t+2t\ln(a_{0}/t)$.
\vskip 0.2cm

To prove this Claim, we first prove the following inequality.

$\displaystyle \sum_{i=0}^{g-1}|X_{i+1}|$
\vskip -0.5cm
\begin{equation}
\leq\left\{\begin{array}{ll} 2t\ln\frac{a_0}{a_q}+2t+2t\ln\frac{a_q+6t}t, & \mbox{if}\ a_q\geq2t,\\
 4t-a_q+2t\ln\frac{a_0}{2t}+2t\ln\frac{a_q+6t}t, & \mbox{if}\ -5t\leq a_q<2t,\\
 14t+a_q+2t\ln\frac{a_0}{2t},& \mbox{if}\ a_q<-5t .\end{array}\right.\label{eq15-1-16-2}
\end{equation}
%\displaystyle

The sequence $a_{1},a_{2},\ldots,a_{g}$ is monotone decreasing with respect to $i$ and the function $\min\{1,\frac{2t}{x}\}$ is monotone decreasing with respect to $x$. Therefore, if $a_{0}\geq2t$, then by Claim 2, we can estimate $\sum_{i=0}^{q-1}|X_{i+1}|$ by an integral as follows:
\begin{align}
\sum_{i=0}^{q-1}|X_{i+1}|
&\leq\int_{a_{q}}^{a_{0}}\min\{1,\frac{2t}{x}\}dx\notag\\
&=\left\{\begin{array}{ll} \displaystyle 2t\int_{a_q}^{a_{0}}\frac{1}{x}dx, & \mbox{if}\ a_{q}\geq2t,\\ \displaystyle\int_{a_{q}}^{2t}1dx+2t\int_{2t}^{a_{0}}\frac{1}{x}dx , & \mbox{if}\ a_{q}<2t,\end{array}\right.\notag\\
&=\left\{\begin{array}{ll} \displaystyle 2t\ln(a_{0}/a_{q}), & \mbox{if}\ a_{q}\geq2t,\\ \displaystyle 2t-a_{q}+2t\ln(a_{0}/2t), & \mbox{if}\ a_{q}<2t.\end{array}\right.\label{eq15-1-4-7}
\end{align}

Similar argument yields,
$$
\sum_{i=q}^{g-1}|X_{i+1}|\leq\left\{\begin{array}{ll} 2(t-b_{g})+ 2t\ln(b_{q}/t), & \mbox{if}\ b_{q}\geq t,\\ 2(b_{q}-b_{g}), & \mbox{if}\ b_{q}<t.\end{array}\right.
$$
Notice that $b_{g}=0$ and $b_{q}=a_{q}+6t$. So
\begin{align}
\sum_{i=q}^{g-1}|X_{i+1}|&\leq\left\{\begin{array}{ll} 2t+ 2t\ln((a_{q}+6t)/t), & \mbox{if}\ a_{q}\geq -5t,\\ 2(a_{q}+6t), & \mbox{if}\ a_{q}<-5t.\end{array}\right.\label{eq15-1-4-14}
\end{align}

Combining (\ref{eq15-1-4-7}) and (\ref{eq15-1-4-14}), inequality (\ref{eq15-1-16-2}) follows.

Next, we estimate the right hand side of (\ref{eq15-1-16-2}). If $a_{q}\geq2t$, then
\begin{align}
&\ln\left(\frac{a_{0}}{a_{q}}\right)+\ln\left(\frac{a_{q}+6t}{t}\right) \notag\\
=&\ln\left(\frac{a_{0}(a_{q}+6t)}{a_{q}t}\right)
=\ln\left(\frac{a_{0}}{t}\right)+\ln\left(\frac{a_{q}+6t}{a_{q}}\right)\notag\\
=&\ln\left(\frac{a_{0}}{t}\right)+\ln\left(1+\frac{6t}{a_{q}}\right)
\leq\ln\left(\frac{a_{0}}{t}\right)+\ln4.\label{eq15-1-10-5}
\end{align}
So in this case, $\sum_{i=0}^{g-1}|X_{i+1}|\leq2t
+2t\ln4+2t\ln(a_{0}/t)<4.78t+2t\ln(a_{0}/t)$. It is easy to see that
\begin{equation}\label{eq15-1-14-1}
\begin{array}{c}
\mbox{when $z=-4t$, function $-z+2t\ln((z+6t)/t)$}\\
\mbox{achieves its maximum value $4t+2t\ln2$.}
\end{array}
\end{equation}
So in the case $-5t\leq a_{q}<2t$, we have $\sum_{i=0}^{g-1}\leq8t+2t\ln(a_{0}/t)$. If $a_q<-5t$, then $a+b\leq 9t+2t\ln (a_0/2t)<7.62t+2t\ln (a_0/t)$. In any case, Claim 3 is true.
%It follows that
%\begin{align}
%\sum_{i=0}^{g-1}|X_{i+1}|
%&=\sum_{i=0}^{q-1}|X_{i+1}|+\sum_{i=q}^{g-1}|X_{i+1}|\notag\\
%&\leq \sum_{i=0}^{q-1}\min\left\{1,\frac{2t}{a_{i}}\right\}\cdot(a_{i}-a_{i+1})+\sum_{i=q}^{g-1}
%\min\{2,\frac{2t}{b_{i}}\}\cdot(b_{i}-b_{i+1}).\label{eq15-1-4-6}
%\end{align}
%For simplicity, denote %$a=\sum_{i=0}^{q-1}
%\min\{1,\frac{2t}{a_{i}}\}\cdot(a_{i}-a_{i+1})$ and $b=\sum_{i=q}^{g-1}
%\min\{2,\frac{2t}{b_{i}}\}\cdot(b_{i}-b_{i+1})$.

\vskip 0.2cm \noindent {\em Claim 4.} If $a_{0}<2t$, then $\sum_{i=0}^{g-1}|X_{i+1}|\leq a_{0}+6t+2t\ln2.$
\vskip 0.2cm

If $a_{0}<2t$, then $a_{i}<2t$ for any $i=0,1,\ldots,g-1$. In this case, by Claim 2, $\sum_{i=0}^{g-1}|X_{i+1}|$ can be estimated as
\begin{align*}
\sum_{i=0}^{g-1}|X_{i+1}|
&=\sum_{i=0}^{q-1}|X_{i+1}|+\sum_{i=q}^{g-1}|X_{i+1}|\notag\\
&\leq \sum_{i=0}^{q-1}(a_{i}-a_{i+1})+\sum_{i=q}^{g-1}
\min\{2,\frac{2t}{b_{i}}\}\cdot(b_{i}-b_{i+1})\notag\\
&\leq\left\{\begin{array}{ll} a_{0}\!-\!a_{q}\!\!+\!\!2t\!\!+\!\!2t\ln\frac{a_{q}+6t}t, & \mbox{if}\ a_{q}\geq -5t,\\
a_0\!-\!a_{q}\!+\!2a_{q}+12t, & \mbox{if}\ a_{q}<-5t ,\end{array}\right.
\end{align*}
If $a_q\geq -5t$, then by using \eqref{eq15-1-14-1}, we have $\sum_{i=0}^{g-1}|X_{i+1}|\leq a_{0}+6t+2t\ln2$. If $a_q<-5t$, then $\sum_{i=0}^{g-1}|X_{i+1}|\leq a_0+a_q+12t\leq a_{0}+7t$. In any case, Claim 4 is true.

\vskip 0.2cm \noindent {\em Claim 5.} For any 2-connected graph $H$, $f(H)\leq2|V(H)|-5.$
\vskip 0.2cm

We prove the Claim by induction on the number of nodes of $H$. If $|V(H)|=3$, then $H$ is a triangle and $f(H)=2|V(H)|-5$. Suppose the Claim is true when $V(H)=n-1$. Consider the case $|V(H)|=n$. If $H$ is a cycle or 3-connected, then by the definition of potential function $f$, we have $f(C_0)=2n-5$ or $f(C_0)=1\leq2n-5$. Otherwise, let $S$ be a good 2-separator and $C_{1}^{S},C_{2}^{S},\ldots,C_{l}^{S}$ be the marked $S$-components of $H$. By Lemma \ref{lemma1},  $C_{1}^{S},C_{2}^{S},\ldots,C_{l}^{S}$ are 2-connected and thus $f(C_{i}^{S})\leq2|C_{i}^{S}|-5$ for $1\leq i\leq l$ by induction hypothesis. Since $\sum_{i=1}^{l}|C_{i}^{S}|=|V(H)|+2l-2$ and any brick of $H$ is completely contained in some $C_{i}^{S}$,
so $f(H)=\sum_{i=1}^{l}f(C_{i}^{S})\leq\sum_{i=1}^{l}(2|C_{i}^{S}|-5)
=2|V(H)|-l-4\leq2|V(H)|-5.$ Claim 5 is proved.

Combing Claim 1 and Claim 5,
\begin{equation}\label{eq13-1-14-4}
a_{0}=f(C_{0})-1-6t<(2\alpha-6)t.
\end{equation}
So, if $\alpha<4$, then $a_{0}<2t$. By Claim 4 and inequality (\ref{eq13-1-14-4}), $\sum_{i=1}^{g}|X_{i}|\leq(2\alpha+2\ln2)t$. If $\alpha\geq4$, then $2\alpha-6\geq 2$. We see from Claim 3, Claim 4, and inequality (\ref{eq13-1-14-4}) that $\sum_{i=1}^{g}|X_{i}|\leq(8+2\ln(2\alpha-6))t$ holds no matter whether $a_0\geq 2t$ or $a_0<2t$.

Combining the above analysis with Claim 1 and the fact
$$|C_{g}|=|C_{0}|+|X_{1}|+\ldots+|X_{g}|,$$
we see that $C_g$ is a $\gamma$-approximation.
\end{proof}

\section{Conclusion}\label{sec5}
In this paper, we have presented a polynomial-time $\gamma$-approximation algorithm for the minimum $(3,m)$-CDS problem, $\gamma=\alpha+8+2\ln(2\alpha-6)$ for $\alpha\geq4$ and $\gamma=3\alpha+2\ln2$ for $\alpha<4$, where $\alpha$ is the approximation ratio for the minimum $(2,m)$-CDS problem. This is the first performance guaranteed approximation algorithm for minimum $(3,m)$-CDS on a general graph and also gives a big improvement on performance ratio of previously known approximation algorithms on unit disk graphs.

For future studies, a natural question is whether the general $(k,m)$-CDS problem also admits an approximation within factor $\ln \delta+o(\ln\delta)$. Recently, CDS considering routing-cost has been studies extensively \cite{Ding,DuH1,DuH,Zhang3}. However, nothing has been done on fault-tolerant issue. This is also a direction for our further research.

\section*{Acknowledgment}
This research is supported by NSFC (11771013,11531011,61751303).

\end{document}